\documentclass[envcountsame, runningheads, a4paper]{llncs}

\usepackage[dvipsnames]{xcolor}
\usepackage{amsmath,amsfonts,amssymb}
\usepackage[utf8]{inputenc}
\usepackage{graphicx}
\usepackage{tikz}
\usepackage{hyperref}
\usepackage[backgroundcolor=magenta!10!white]{todonotes}
\usepackage{setspace}
\usepackage{microtype}
\usepackage{enumitem}
\usepackage{ifthen}
\usepackage{caption}
\usepackage{microtype}
\usepackage[space]{cite}
\usepackage{numprint}
\usepackage{cleveref}
\usepackage{multirow}
\usepackage{subcaption}
\usepackage{placeins}
\usepackage{rotating}
\usepackage{tikz-qtree}
\usepackage{thmtools}
\usepackage{marvosym}

\usetikzlibrary{arrows,
automata,
calc,
decorations,
decorations.pathreplacing,
positioning,
shapes, 
decorations.pathmorphing, 
matrix}
\usepackage{academicons}
\usepackage{textcomp}
\usepackage{stmaryrd}
\usepackage{enumitem}
\usepackage{tikz}
\usepackage{orcidlink}
\usetikzlibrary{arrows, automata, shapes, decorations.pathmorphing, calc, positioning, matrix}
\usepackage[symbol]{footmisc}

\newcommand{\logic}{\cal L}
\newcommand{\reachvar}{\mathbf{r}}
\newcommand{\instantiate}{\operatorname{Instantiate}}

\newcommand{\enf}{\operatorname{Enf}}
\newcommand{\pre}{\operatorname{Pre}}
\newcommand{\post}{\operatorname{Post}}
\newcommand{\G}{\mathcal{G}}
\newcommand{\T}{\mathit{T}}
\newcommand{\init}{\mathit{Init}}
\newcommand{\safe}{\mathit{Safe}}
\newcommand{\reach}{\mathit{Reach}}
\newcommand{\goal}{\mathit{Goal}}
\newcommand{\thief}{\mathit{Thief}}
\newcommand{\guard}{\mathit{Guard}}
\newcommand{\var}{\mathcal{V}}

\newcommand{\varp}{\mathcal{V'}}

\newcommand{\subsp}{[\var/\varp]}
\newcommand{\subsm}{[\varp/\var]}
\newcommand{\smallv}{\mathit{v}}
\newcommand{\smallvp}{\mathit{v'}}
\newcommand{\game}{$\G = \langle \init, \safe, \reach, \goal \rangle$}
\newcommand{\sat}{\operatorname{Sat}}
\newcommand{\unsat}{\operatorname{Unsat}}
\newcommand{\interpol}{\operatorname{Interpolate}}
\newcommand{\algreach}{\operatorname{Reach}}

\newcommand{\plays}{\operatorname{Plays}}

\newcommand{\playerR}{\texttt{{REACH}}}
\newcommand{\playerS}{\texttt{{SAFE}}}
\newcommand{\stratR}{\sigma_{\mathit{R}}}
\newcommand{\stratS}{\sigma_{\mathit{S}}}
\newcommand{\simsynth}{\textsc{SimSynth}}
\newcommand{\cabpy}{\textsc{CabPy}}
\newcommand{\domain}{\mathcal{D}}
\newcommand{\true}{\texttt{true}}
\newcommand{\false}{\texttt{false}}

\newcommand{\len}{\operatorname{len}}

\newcommand{\size}{\operatorname{size}}

\newcommand{\combine}{\operatorname{combine}}

\renewcommand{\S}{\mathit{S}}

\newcommand{\sstrat}{\mathfrak{S}}

\newcommand{\interpolate}{\operatorname{Interpolate}}

\usepackage{caption}

\usepackage[ruled,vlined,english,linesnumbered]{algorithm2e}
\makeatletter
\newcommand{\RemoveAlgoNumber}{\renewcommand{\fnum@algocf}{\AlCapSty{\AlCapFnt\algorithmcfname}}}
\newcommand{\RevertAlgoNumber}{\algocf@resetfnum}
\makeatother

\colorlet{christelColor}{Apricot!30!white}
\colorlet{simonColor}{Yellow!30!white}
\colorlet{florianColor}{Brown!30!white}
\colorlet{julianColor}{Blue!20!white}
\colorlet{norineColor}{Green!20!white}

\newcounter{flocomments}

 \newcounter{juliancomments}

\newcommand{\orcid}[1]{\href{https://orcid.org/#1}{\textcolor[HTML]{A6CE39}{\aiOrcid}}}

\title{Causality-Based Game Solving
\thanks{This work was partially supported by DFG grant 389792660 as part of TRR~248 -- CPEC, see \url{https://perspicuous-computing.science}, the Cluster of Excellence EXC 2050/1 (CeTI, project ID 390696704, as part of Germany’s Excellence Strategy), DFG-projects BA-1679/11-1 and BA-1679/12-1, the Research Training Group QuantLA (GRK 1763), and by the European Research Council (ERC) Grant OSARES (No. 683300).}}

\author{Christel Baier\inst{1}\orcidID{0000-0002-5321-9343} \and 
	Norine Coenen\inst{2}\orcidID{0000-0003-2066-1511} \and\\ 
	Bernd Finkbeiner\inst{2}\orcidID{0000-0002-4280-8441} \and
	Florian Funke\inst{1}\orcidID{0000-0001-7301-1550}\and Simon Jantsch\inst{1}\orcidID{0000-0003-1692-2408} \and 
	Julian Siber\inst{2}\textsuperscript{(\Letter)}\orcidID{0000-0003-0842-0029}}
\authorrunning{C. Baier et al.}
\institute{
	Technische Universität Dresden, Dresden, Germany 
	\email{\{christel.baier,florian.funke,simon.jantsch\}@tu-dresden.de}\and
	CISPA Helmholtz Center for Information Security, Saarbrücken, Germany
	\email{\{norine.coenen,finkbeiner,julian.siber\}@cispa.de}
}

\begin{document}
\maketitle

\begin{abstract} 
	We present a causality-based algorithm for solving two-player reachability games represented by logical constraints.
	These games are a useful formalism to model a wide array of problems arising, e.g., in program synthesis. 
	Our technique for solving these games is based on the notion of \emph{subgoals}, which are slices of the game that the reachability player necessarily needs to pass through in order to reach the goal. 
	We use Craig interpolation to identify these necessary sets of moves and recursively slice the game along these subgoals. 
	Our approach allows us to infer winning strategies that are structured along the subgoals. 
	If the game is won by the reachability player, this is a strategy that progresses through the subgoals towards the final goal; if the game is won by the safety player, it is a permissive strategy that completely avoids a single subgoal. 
	We evaluate our prototype implementation on a range of different games. 
	On multiple benchmark families, our prototype scales dramatically better than previously available tools. 
\end{abstract}

\section{Introduction}
Two-player games are a fundamental model in logic and verification due to their connection to a wide range of topics such as decision procedures, synthesis and control~\cite{Harding05,Alur15,Alur16,BLOEM2012911,BLOEM20073,6225075}. 
Algorithmic techniques for \emph{finite-state} two-player games have been studied extensively for many acceptance conditions~\cite{GraedelTW2002}.
For \emph{infinite-state} games most problems are directly undecidable. 
However, infinite state spaces occur naturally in domains like software synthesis~\cite{RyzhykCKLH09} and cyber-physical systems~\cite{JessenRLD07}, and hence handling such games is of great interest. 
An elegant classification of infinite-state games that can be algorithmically handled, depending on the acceptance condition of the game, was given in~\cite{deAlfaroHM2001}.
The authors assume a symbolic encoding of the game in a very general form.
More recently, incomplete procedures for solving infinite-state two-player games specified using logical constraints were studied~\cite{BeyeneCPR14,FarzanK17}.
While~\cite{BeyeneCPR14} is based on automated theorem-proving for Horn formulas and handles a wide class of acceptance conditions, the work in~\cite{FarzanK17} focusses on reachability games specified in the theory of linear arithmetic, and uses sophisticated decision procedures for that theory.

In this paper, we present a novel technique for solving logically represented reachability games based on the notion of \emph{subgoals}.
A \emph{necessary} subgoal is a transition predicate that is satisfied at least once on every play that reaches the overall goal.
It represents an intermediate target that the reachability player must reach in order to win. Subgoals open up game solving to the study of cause-effect relationships in the form of counterfactual reasoning~\cite{counterfactual}: If a cause (the subgoal) had not occurred, then the effect (reaching the goal) would not have happened.
Thus for the safety player, a necessary subgoal provides a chance to win the game based on local information:
If they control all states satisfying the pre-condition of the subgoal, then any strategy that in these states picks a transition outside of the subgoal is  winning. 
Finding such a necessary subgoal may let us conclude that the safety player wins without ever having to unroll the transition relation.

On the other hand, passing through a necessary subgoal is in general not enough for the reachability player to win. We call a subgoal \emph{sufficient} if indeed the reachability player has a winning strategy from every state satisfying the post-condition of the subgoal.
Dual to the description in the preceding paragraph, sufficient subgoals provide a chance for the reachability player to win the global game as they must merely reach this intermediate target. The two properties differ in one key aspect: While necessity of a subgoal only considers the paths of the game arena, for sufficiency the game structure is crucial.  

We show how Craig interpolants can be used to compute necessary subgoals, making our methods applicable to games represented by any logic that supports interpolation. In contrast, determining whether a subgoal is sufficient requires a partial solution of the given game. This motivates the following recursive approach. We slice the game along a necessary subgoal into two parts, the pre-game and the post-game.
In order to guarantee these games to be smaller, we solve the post-game under the assumption that the considered subgoal was bridged \emph{for the last time}.
We conclude that the safety player wins the overall game if they can avoid all initial states of the post-game that are winning for the reachability player.
Otherwise, the pre-game is solved subject to the winning condition given by the sufficient subgoal consisting of these states. 
This approach does not only determine which player wins from each initial state, but also computes symbolically represented winning strategies with a causal structure. 
Winning safety player strategies induce necessary subgoals that the reachability player cannot pass, which constitutes a cause for their loss. Winning reachability player strategies represent a sequence of sufficient subgoals that will be passed, providing an explanation for the win.

The Python-based implementation \cabpy{} of our approach was used to compare its performance to \simsynth \cite{FarzanK17}, which is, to the best of our knowledge, the only other available tool for solving linear arithmetic reachability games. Our experiments demonstrate that our algorithm is competitive in many case studies. We can also confirm the expectation that our approach heavily benefits from qualitatively expressive Craig interpolants. It is noteworthy that like \simsynth{} our approach is fully automated and does not require any input in the form of hints or templates. 
Our contributions are summarized as follows:
\begin{itemize}[topsep=0.5ex]
	\item We introduce the concept of \emph{necessary} and \emph{sufficient subgoals} and show how Craig interpolation can be used to compute necessary subgoals (\Cref{sec:subgoals}).
	\item We describe an algorithm for solving logically represented two-player reachability games using these concepts.
    We also discuss how to compute representations of winning strategies in our approach (\Cref{sec:gamesolving}).
	\item We evaluate our approach experimentally through our Python-based tool \cabpy, demonstrating a competitive performance compared to the previously available tool \simsynth{} on various case studies (\Cref{sec:experiments}).
\end{itemize}

\textbf{Related Work. } 
The problem of solving linear arithmetic games is addressed in~\cite{FarzanK17} using an approach that relies on a dedicated decision procedure for quantified linear arithmetic formulas, together with a method to generalize safety strategies from truncated versions of the game  that end after a prescribed number of rounds.
Other approaches for solving infinite-state games include deductive methods that compute the winning regions of both players using proof rules~\cite{BeyeneCPR14}, 
predicate abstraction where an abstract controlled predecessor operation is used on the abstract game representation~\cite{WalkerR14}, 
and symbolic BDD-based exploration of the state space~\cite{Edelkamp2002}. 
Additional techniques are available for finite-state games, e.g., generalizing winning runs into a winning strategy for one of the players~\cite{NarodytskaLBRW14}. 

Our notion of subgoal is related to the concept of landmarks as used in planning~\cite{HoffmannPS11}. Landmarks are milestones that must be true on every successful plan, and they can be used to decompose a planning task into smaller sub-tasks. 
Landmarks have also been used in a game setting to prevent the opponent from reaching their goal using counter-planning~\cite{PozancoEFB18}. 
Whenever a planning task is unsolvable, one method to find out why is checking hierarchical abstractions for solvability and finding the components causing the problem~\cite{SreedharanSSK19}. 

Causality-based approaches have also been used for model checking of multi-threaded concurrent programs~\cite{DBLP:conf/concur/KupriyanovF13,DBLP:conf/cav/KupriyanovF14}. 
In our approach, we use Craig interpolation to compute the subgoals. 
Interpolation has already been used in similar contexts before, for example to extract winning strategies from game trees~\cite{EenLNR15} or to compute new predicates to refine the game abstractions~\cite{SlicingAbstractions}. 
In ~\cite{FarzanK17}, interpolation is used to synthesize concrete winning strategies from so called \emph{winning strategy skeletons}, which describe a set of strategies of which at least one is winning.

\section{Motivating Example} \label{sec:motivation}

Consider the scenario that an expensive painting is displayed in a large exhibition room of a museum.
It is secured with an alarm system that is controlled via a control panel on the opposite side of the room.
A security guard is sleeping at the control panel and occasionally wakes up to check whether the alarm is still armed.
To steal the painting, a thief first needs to disable the alarm and then reach the painting before the alarm has been reactivated. We model this scenario as a two-player game between a safety player (the guard) and a reachability player (the thief) in the theory of linear arithmetic.
The moves of both players, their initial positions, and the goal condition are described by the formulas:

\begin{align*}
\init &\equiv && \lnot \reachvar \land x = 0 \land y = 0 \land p = 0 \land a = 1 \land t = 0, &&&\\
\guard &\equiv &&\neg\reachvar \land \reachvar' \land x' = x \land y' = y \land p' = p  &&&\\
& &&\land ((t' = t - 1 \land a' = a)\lor (t \leq 0 \land t' = 2)),&&&(\text{sleep or wake up})\\
\thief &\equiv &&  \reachvar \land \neg\reachvar' \land t' = t \\
	&&&\land x + 1 \ge x' \ge x - 1 \land y + 1 \ge y' \ge y - 1&&&(\text{move})\\
& &&\land (x' \neq 0 \lor y' \neq 10 \implies a' = a)&&&(\text{alarm off})\\
& &&\land (x' \neq 10 \lor y' \neq 5 \lor a = 1 \implies p' = p),&&&(\text{steal})\\
\goal &\equiv && \lnot \reachvar \land p = 1. &&&
\end{align*} 

The thief's position in the room is modeled by two coordinates $x,y \in \mathbb{R}$ with initial value $(0,0)$, and with every transition the thief can move some bounded distance. 
Note that we use primed variables to represent the value of variables after taking a transition.
The control panel is located at $(0,10)$ and the painting at $(10,5)$. 
The status of the alarm and the painting are described by two boolean variables $a,p \in \{0,1\}$. 
The guard wakes up every two time units, modeled by the variable $t \in \mathbb{R}$. 
The variables $x,y$ are bounded to the interval $[0,10]$ and $t$ to $[0,2]$. 
The boolean variable \textbf{r} encodes who makes the next move. 
In the presented configuration, the thief needs more time to move from the control panel to the painting than the guard will sleep. 
It follows that there is a winning strategy for the guard, namely, to always reactivate the alarm upon waking up.

Although it is intuitively fairly easy to come up with this strategy for the guard, it is surprisingly hard for game solving tools to find it. The main obstacle is the infinite state space of this game.
Our approach for solving games represented in this logical way imitates \emph{causal reasoning}: 
Humans observe that in order for the thief to steal the painting (i.e., the effect $p=1$), a transition must have been taken whose source state does not satisfy the pre-condition of (steal) while the target state does.  
Part of this cause is the condition $a=0$, i.e., the alarm is off. Recursively, in order for the effect $a=0$ to happen, a transition setting $a$ from $1$ to $0$ must have occurred, and so on. 

Our approach captures these cause-effect relationships through the notion of \emph{necessary subgoals}, which are essential milestones that the reachability player has to transition through in order to achieve their goal.
The first necessary subgoal corresponding to the intuitive description above is
$$C_1 = (\guard \lor \thief) \land p \neq 1 \land p' = 1.$$
In this case, it easy to see that $C_1$ is also a \emph{sufficient subgoal}, meaning that all successor states of $C_1$ are winning for the thief. Therefore, it is enough to solve the game with the modified objective to reach those predecessor states of $C_1$ from which the thief can \emph{enforce} $C_1$ being the next move (even if it is not their turn). Doing so recursively produces the necessary subgoal
$$C_2 = (\guard \lor \thief) \land a \neq 0 \land a' = 0,$$
meaning that some transition must have caused the effect that the alarm is disabled. However, $C_2$ is \emph{not} sufficient which can be seen by recursively solving the game spanning from successor states of $C_2$ to $C_1$. This computation has an important caveat: After passing through $C_2$, it may happen that $a$ is reset to $1$ at a later point (in this particular case, this constitutes precisely the winning strategy of the safety player), which means that there is no canonical way to slice the game along this subgoal into smaller parts. Hence the recursive call solves the game from $C_2$ to $C_1$ \emph{subject to} the bold assumption that any move from $a = 0$ to $a' = 1$ is winning for the guard. This generally underapproximates the winning states of the thief. Remarkably, we show that this approximation is enough to build winning strategies for \emph{both} players from their respective winning regions. In this case, it allows us to infer that moving through $C_2$ is always a losing move for the thief. However, at the same time, any play reaching $\goal$ has to move through $C_2$. It follows that the thief loses the global game.

We evaluated our method on several configurations of this game, which we call \emph{Mona Lisa}. The results in Section \ref{sec:experiments} support our conjecture that the room size has little influence on the time our technique needs to solve the game.

\section{Preliminaries}
\label{sec:prelims}
We consider two-player reachability games defined by formulas in a given logic $\logic$.
We let $\logic(\var)$ be the $\logic$-formulas over a finite set of variables $\var$, also called \emph{state predicates} in the following.
We call $\varp = \{\smallvp \mid \smallv \in \var\}$ the set of \emph{primed variables}, which are used to represent the value of variables after taking a transition.
Transitions are expressed by formulas in the set $\logic(\var \cup \varp)$, called \emph{transition predicates}.
For some formula $\varphi \in \logic(\var)$, we denote the substitution of all variables by their primed variant by $\varphi\subsp$. Similarly, we define $\varphi\subsm$.

For our algorithm we will require the satisfiability problem of $\logic$-formulas to be decidable and \emph{Craig interpolants} \cite{Craig1957} to exist for any two mutually unsatisfiable formulas.
Formally, we assume there is a function $\sat : \logic(\var) \to \mathbb{B}$ that checks the satisfiability of some formula $\varphi \in \logic(\var)$ and an unsatisfiability check $\unsat : \logic(\var) \rightarrow \mathbb{B}$. 
For interpolation, we assume that there is a function $\interpol : \logic(\var) \times \logic(\var) \to \logic(\var)$ computing a \emph{Craig interpolant} for mutually unsatisfiable formulas: If $\varphi ,\psi \in \logic(\var)$ are such that $\unsat(\varphi\land\psi)$ holds, then $\psi \implies \interpol(\varphi,\psi)$ is valid, $\interpol(\varphi,\psi)\land \varphi$ is unsatisfiable, and $\interpol(\varphi,\psi)$ only contains variables shared by $\varphi$ and $\psi$.

These functions are provided by many modern \emph{Satisfiability Modulo Theories} (SMT) solvers, in particular for the theories of linear integer arithmetic and linear real arithmetic, which we will use for all our examples. Note that interpolation is usually only supported for the quantifier-free fragments of these logics, while our algorithm will introduce existential quantifiers.
Therefore, we resort to quantifier elimination wherever necessary, for which there are known procedures for both linear integer arithmetic and linear real arithmetic formulas~\cite{presburger1929uber,Monniaux2008}.

In order to distinguish the two players, we will assume that a Boolean variable called $\reachvar \in \var$ exists, which holds exactly in the states controlled by the reachability player.
For all other variables $v \in \var$, we let $\domain(v)$ be the domain of $v$, and we define $\domain = \bigcup \{\domain(v) \mid v \in \var\}$.
In the remainder of the paper, we consider the variables $\var$ and their domains to be fixed.

\begin{definition}[Reachability Game]
	A reachability game is defined by a tuple \game{}, where $\init \in \logic(\var)$ is the \emph{initial condition}, $\safe \in \logic(\var \cup \varp)$ defines the transitions of player $\playerS$, $\reach \in \logic(\var \cup \varp)$ defines the transitions of player $\playerR$ and $\goal \in \logic(\var)$ is the \emph{goal condition}.

  We require the formulas $(\safe \implies \neg \reachvar)$ and $(\reach \implies \reachvar)$ to be valid.
\end{definition}

A \emph{state} $s$ of $\G$ is a valuation of the variables $\var$, i.e., a function $s\colon\var \to \domain$ that satisfies $s(v) \in \domain(v)$ for all $v \in \var$.
We denote the set of states by $S$, and we let $S_\playerS$ be the states $s$ such that $s(\reachvar) = \false$, and $S_\playerR$ be the states $s$ such that $s(\reachvar) = \true$. The variable $\reachvar$ determines whether \playerR{} or \playerS{} makes the move out of the current state, and in particular $\safe \land \reach$ is unsatisfiable. 

Given a state predicate $\varphi \in \logic(\var)$, we denote by $\varphi(s)$ the closed formula we get by replacing each occurrence of variable $v \in \var$ in $\varphi$ by $s(v)$.
Similarly, given a transition predicate $\tau \in \logic(\var \cup \varp)$ and states $s,s'$, we let $\tau(s,s')$ be the formula we obtain by replacing all occurrences of $v \in \var$ in $\tau$ by $s(v)$, and all occurrences of $v' \in \varp$ in $\tau$ by $s'(v)$. For replacing only $v \in \var$ by $s(v)$, we define $\tau(s)\in\logic(\var')$. A \emph{trap state} of $\G$ is a state $s$ such that $(\safe \lor \reach)(s)\in\logic(\var')$ is unsatisfiable (i.e., $s$ has no outgoing transitions).

A \emph{play} of $\G$ starting in state $s_0$ is a finite or infinite sequence of states $\rho = s_0 s_1 s_2 \ldots \in \S^+ \cup \S^\omega$ such that for all $i < \len(\rho)$ either $\safe(s_i,s_{i+1})$ or $\reach(s_i,s_{i+1})$ is valid, and if $\rho$ is a finite play, then $s_{\len(\rho)}$ is required to be a trap state.
Here, $\len(s_0\ldots s_n) = n$ for finite plays, and $\len(\rho) = \infty$ if $\rho$ is an infinite play. The set of plays of some game \game{} is defined as $\plays(\G) = \{\rho = s_0 s_1 s_2 \ldots \mid \rho\text{ is a play in } \G \text{ s.t. } \init(s_0)\text{ holds} \}$.
$\playerR$ \emph{wins} some play $\rho = s_0 s_1 \ldots$ if the play reaches a goal state, i.e., if there exists some integer $0\leq k \leq\len(\rho)$ such that $\goal(s_k)$ is valid.
Otherwise, $\playerS$ wins play $\rho$.
A \emph{reachability strategy} $\stratR$ is a function $\stratR : \S^*S_\playerR \to \S$ such that if $\stratR(\omega s) =s'$ and $s$ is not a trap state, then $\reach(s,s')$ is valid. 
We say that a play $\rho = s_0 s_1 s_2 \ldots$ is \emph{consistent} with $\stratR$ if for all $i$ such that $s_i(\reachvar) = \true$ we have $s_{i+1} = \stratR(s_0 \ldots s_i)$.
A reachability strategy $\stratR$ is \emph{winning} from some state $s$ if $\playerR$ wins every play consistent with $\stratR$ starting in $s$. We define \emph{safety strategies} $\stratS$ for $\playerS$ analogously. We say that a player \emph{wins in or from a state}~$s$ if they have a winning strategy from $s$. Lastly, $\playerR$ \emph{wins the game} $\G$ if they win from some initial state.
Otherwise, $\playerS$ wins.

We often project a transition predicate $T$ onto the source or target states of transitions satisfying $T$, which is taken care of by the formulas $\pre(\T)=\exists \varp.\:\T$ and $\post(\T)=\exists \var.\:\T$.
The notation $\exists \var$ (resp. $\exists \varp$) represents the existential quantification over all variables in the corresponding set.
Given $\varphi \in \logic(\var)$, we call the set of transitions in $\G$ that move from states not satisfying $\varphi$, to states satisfying $\varphi$, the \emph{instantiation} of $\varphi$, formally:
\[\instantiate(\varphi,\G)=(\safe \lor \reach) \land \lnot\varphi\land\varphi'.\]

\section{Subgoals}
\label{sec:subgoals}

We formally define the notion of subgoals.
Let $\G = \langle \init, \safe, \reach, \goal \rangle$ be a fixed reachability game throughout this section, where we assume that $\init \land \goal$ is unsatisfiable.
Whenever this assumption is not satisfied in our algorithm, we will instead consider the game $\G' = \langle \init \land \neg \goal, \safe, \reach, \goal \rangle$ which does satisfy it.
As states in $\init \land \goal$ are immediately winning for \playerR, this is not a real restriction.

\begin{definition}[Enforceable transitions]
	The set of \emph{enforceable transitions} relative to a transition predicate $T \in \logic(\var \cup \varp)$ is defined by the formula
	\[\enf(\T,\G)=\; (\safe \lor \reach) \land \T \land \lnot \exists \varp.\:\big(\safe\land\lnot\T\big).\]
	
\end{definition}

The enforceable transitions operator serves a purpose similar to the \emph{controlled predecessors} operator commonly known in the literature, which is often used in a backwards fixed point computation, called \emph{attractor construction} \cite{Thomas95}. For both operations, the idea is to determine controllability by \playerR{}. The main difference is that we do not consider the whole transition relation, but only a predetermined set of transitions and check from which predecessor states the post-condition of the set can be enforced by \playerR{}. These include all transitions in $T$ controlled by \playerR{} and additionally transitions in $T$ controlled by \playerS{} such that \emph{all other transitions} in the origin state of the transition also satisfy $T$. The similarity with the controlled predecessor is exemplified by the following lemma:
\begin{lemma}
  \label{lem:enf}
  Let $T$ be a transition predicate, and suppose that all states satisfying $\post(T)\subsm$ are winning for \playerR{} in $\G$.
  Then all states in $\pre(\enf(T,\G))$ are winning for \playerR{} in $\G$.
\end{lemma}
\begin{proof}
  Clearly, all states in $\pre(\enf(T,\G))$ that are under the control of \playerR{} are winning for \playerR{}, as in any such state they have a transition satisfying $T$ (observe that $\enf(T,\G) \implies T$ is valid), which leads to a winning state by assumption.

  So let $s$ be a state satisfying $\pre(\enf(T,\G))$ that is under the control of \playerS{}.
  As $\pre(\enf(T,\G))(s)$ is valid, $s$ has a transition that satisfies $T$ (in particular, $s$ is not a trap state).
  Furthermore, we know that there is no $s' \in \S$ such that $\safe(s,s')\land\lnot\T(s,s')$ holds, and hence there is no transition satisfying $\lnot\T$ from $s$. Since $\post(T)\subsm$ is winning for \playerR{}, it follows that from $s$ player \playerS{} cannot avoid playing into a winning state of \playerR{}.
  \qed
\end{proof}

We now turn to a formal definition of \emph{necessary subgoals}, which intuitively are sets of transitions that appear on every play that is winning for \playerR{}. 
\begin{definition}[Necessary subgoal]\label{necessary_subgoal}
	A \emph{necessary subgoal} $C \in \logic(\var \cup \varp)$  for~$\G$ is a transition predicate such that for every play $\rho = s_0 s_1 \ldots$ of $\G$ and $n \in \mathbb{N}$ such that $\goal(s_{n})$ is valid, there exists some $k < n$ such that $C(s_k,s_{k+1})$ is valid.
\end{definition}

Necessary subgoals provide a means by which winning safety player strategies can be identified, as formalized in the following lemma.

\begin{lemma}\label{prop_safestrat}
	A safety strategy $\stratS$ is winning in $\G$ if and only if there exists a necessary subgoal $\mathit{C}$ for $\G$ such that for all plays $\rho = s_0 s_1 \ldots$ of $\G$ consistent with~$\stratS$ there is no $n \in \mathbb{N}$ such that $C(s_n,s_{n+1})$ holds. 
\end{lemma}
\begin{proof}
  ``$\implies$''. The transition predicate $\goal\subsp$ (i.e., transitions with endpoints satisfying $\goal$) is clearly a necessary subgoal. If $\stratS$ is winning for \playerS, then no play consistent with $\stratS$ contains a transition in this necessary subgoal. \\
  \noindent ``$\Longleftarrow$''. Let $C$ be a necessary subgoal such that no play consistent with $\stratS$ contains a transition of $C$. Then by  \Cref{necessary_subgoal} no play consistent with $\stratS$ contains a state satisfying $\goal$. Hence $\stratS$ is a winning strategy for \playerS.
	\qed
\end{proof}

Of course, the question remains how to compute non-trivial subgoals. Indeed, using $\goal$ as outlined in the proof above provides no further benefit over a simple backwards exploration (see~\Cref{rem:attractor} in the following section).

Ideally, a subgoal should represent an interesting key decision to focus the strategy search.
As we show next, Craig interpolation allows to extract partial causes for the mutual unsatisfiability of $\init$ and $\goal$ and can in this way provide necessary subgoals. 
Recall that a Craig interpolant $\varphi$ between $\init$ and $\goal$ is a state predicate that is implied by $\goal$, and unsatisfiable in conjunction with $\init$. 
In this sense, $\varphi$ describes an observable \emph{effect} that must occur if $\playerR{}$ wins, and the concrete transition that instantiates the interpolant \emph{causes} this effect.

\begin{proposition}\label{prop_necessary}
	Let $\varphi$ be a Craig interpolant for $\init$ and $\goal$. Then the transition predicate $\instantiate(\varphi,\G)$ is a necessary subgoal.
\end{proposition}
\begin{proof}
  As $\varphi$ is an interpolant, it holds that $\goal \implies \varphi$ is valid and $\init \land \varphi$ is unsatisfiable.
  Consider any play $\rho = s_0 s_1 \ldots$ of $\G$ such that $\goal(s_n)$ is valid for some $n \in \mathbb{N}$.
  It follows that $\lnot \varphi(s_0)$ and $\varphi(s_n)$ are both valid.
  Consequently, there is some $0 \leq i < n$ such that $\lnot \varphi(s_i)$ and $\varphi(s_{i+1})$ are both valid.
  As all pairs $(s_k,s_{k+1})$ satisfy either $\safe$ or $\reach$, it follows that $\big(\instantiate(\varphi,\G)\big)(s_i,s_{i+1})$ is valid.
  Hence, $\instantiate(\varphi,\G)$ is a necessary subgoal.
	\qed
\end{proof}

While avoiding a necessary subgoal is a winning strategy for \playerS{}, reaching a necessary subgoal is in general not sufficient to guarantee a win for \playerR{}.
This is because there might be some transitions in the necessary subgoal that produce the desired effect described by the Craig interpolant, but that trap \playerR{} in a region of the state space where they cannot enforce some other necessary effect to reach goal. 
For the purpose of describing a set of transitions that is guaranteed to be winning for the reachability player, we introduce \emph{sufficient subgoals}.

 \begin{definition}[Sufficient subgoal]
   A transition predicate $\mathit{F}\in\logic(\var\cup\varp)$ is called a \emph{sufficient subgoal} if $\playerR$ wins from every state satisfying $\post(\mathit{F})\subsm$.
 \end{definition}

\begin{example}
	Consider the Mona Lisa game $\G$ described in Section \ref{sec:motivation}.
	\[C_1 = (\guard \lor \thief) \land p \neq 1 \land p' = 1\]
	qualifies as sufficient subgoal, because $\playerR$ wins from every successor state as all those states satisfy $\goal$. 
	Also, every play reaching $\goal$ eventually passes $C_1$, and hence $C_1$ is also necessary. On the other hand, 
	\[C_2 = (\guard \lor \thief) \land a \neq 0 \land a' = 0\]
	is only a necessary subgoal in $\G$, because $\playerS$ wins from some (in fact all) states satisfying $\post(C_2)$.
	
\end{example}
 
If the set of transitions in the necessary subgoal $C$ that lead to winning states of \playerR{} is definable in $\logic$ then we call the transition predicate $F$ that defines it the \emph{largest sufficient subgoal} included in $C$. 
It is characterized by the properties (1) $F \implies C$ is valid, and (2) if $F'$ is such that $F \implies F'$ is valid, then either $F \equiv F'$, or $F'$ is not a sufficient subgoal. Since $C$ is a necessary subgoal and $F$ is maximal with the properties above, \playerR{} needs to see a transition in $F$ eventually in order to win. This balance of necessity and sufficiency allows us to partition the game along $F$ into a game that happens after the subgoal and one that happens before.

\begin{proposition}
		\label{lem:slicing}
	  Let $C$ be a necessary subgoal, and $F$ be the largest sufficient subgoal included in $C$. Then \playerR{} wins from an initial state $s$ in $\G$ if and only if \playerR{} wins from $s$ in the pre-game 
    \[\G_{pre} = \langle \init, \safe \land \neg F, \reach \land \neg F, \pre(\enf(F,\G)) \rangle.\]
\end{proposition}
\begin{proof}
  ``$\implies$''. Suppose that \playerR{} wins in $\G$ from $s$ using strategy $\sigma_R$. Assume for a contradiction that \playerS{} wins in $\G_{pre}$ from $s$ using strategy $\sigma_S$. Consider strategy $\sigma'_S$ such that $\stratS'(\omega s') = \stratS(\omega s')$ if $(\safe \land \lnot F)(s')$ is satisfiable, and else $\stratS'(\omega s') = \stratS''(\omega s')$, where $\stratS''$ is an arbitrary safety player strategy in $\G$. Let $\rho = s_0s_1\ldots$ be the (unique) play of $\G$ consistent with both $\sigma_R$ and $\sigma'_S$, where $s_0 = s$. Since $\sigma_R$ is winning in $\G$ and $C$ is a necessary subgoal in $\G$, there must exist some $m\in\mathbb{N}$ such that $C(s_m, s_{m+1})$ is valid. Let $m$ be the smallest such index. Since $F \implies C$, we know for all $0 \leq k < m$ that $\lnot F (s_k,s_{k+1})$ holds. Hence, there is the play $\rho' = s_0s_1\ldots s_m \ldots$ in $\G_{pre}$ consistent with $\sigma_S$. The state $s_{m+1}$ is winning for \playerR{} in $\G$, as it is reached on a play consistent with the winning strategy $\sigma_R$. Hence, we know that $F(s_m, s_{m+1})$ holds, because $F$ is the largest sufficient subgoal included in $C$.
  If $(\reach \land F)(s_m, s_{m+1})$ held, we would have that $\pre(\enf(F,\G)(s_m)$ holds: a contradiction with $\rho'$ being consistent with $\sigma_S$, which we assumed to be winning in $\G_{pre}$. It follows that $(\safe \land F)(s_m, s_{m+1})$ holds. We can conclude that $(\safe \land \lnot F)(s_m)$ is unsatisfiable (i.e., $s_m$ is a trap state in $\G_{pre}$), because in all other cases $\playerS$ plays according to $\stratS$, which cannot choose a transition satisfying $F$. However, this implies that  $\pre(\enf(F,\G)(s_m)$ holds, again a contradiction with $\rho'$ being consistent with winning strategy $\sigma_S$.

	\noindent ``$\Longleftarrow$''. If $\playerR$ wins in $\G_{pre}$ they have a strategy $\stratR$ such that every play consistent with $\stratR$ reaches the set $\pre(\enf(F,\G))$.
  As $F$ is a sufficient subgoal, the states $\post(F)$ are winning for \playerR{} by definition.
  It follows by~\Cref{lem:enf} that all states satisfying $\pre(\enf(F,\G))$ are winning in $\G$.
  Combining $\sigma_{R}$ with a strategy that wins in all these states yields a winning strategy for \playerR{} in $\G$.
  \qed
\end{proof}

\section{Causality-Based Game Solving}
\label{sec:gamesolving}

\Cref{lem:slicing} in the preceding section foreshadows how subgoals can be employed in building a recursive approach for the solution of reachability games. 
Before turning to our actual algorithm, we describe a way to symbolically represent nondeterministic memoryless strategies. As discussed in \cite{FarzanK17}, there is no ideal strategy description language for the class of games we consider. Our approach allows us to describe sets of concrete strategies as defined in Section~\ref{sec:prelims} with linear arithmetic formulas.
This framework will prove convenient for \emph{strategy synthesis}, i.e., the computation of winning strategies instead of simply determining the winner of the game.

\subsection{Symbolically Represented Strategies}
We will represent strategies for both players using transition predicates $\sstrat \in \logic(\var \cup \varp)$, henceforth called \emph{symbolic strategies}, where we only require that $(\sstrat \implies (\safe \lor \reach))$ is valid.
A sequence $s_0 \ldots s_n \in S^+$ is called a \emph{play prefix} if it is a prefix of some play in $\G$, $(\neg \goal)(s_j)$ holds for all $0 \leq j \leq n$, and $s_n$ is not a trap state.
We say that a play prefix $\rho = s_0 \ldots s_n$ \emph{conforms to} a symbolic reachability strategy $\sstrat$ if for all $j < n$ we have that $\sstrat(s_j,s_{j+1})$ holds whenever $s_j \in S_{\playerR}$ (and analogously for safety strategies).
A play conforms to~$\sstrat$ if all its play prefixes conform to $\sstrat$.
We say that $\sstrat$ is winning for \playerR{} in $s$ if all plays from $s$ that conform to $\sstrat$ are winning for \playerR{}  and all play prefixes $s_0 \ldots s_n \in S^*S_{\playerR{}}$ from $s$ that conform to $\sstrat$ are such that $(\sstrat \land \reach)(s_n)$ is satisfiable (and analogously for \playerS{}).
The second condition ensures that the player cannot be forced to play a transition outside of $\sstrat$ by their opponent while the play has not reached a trap state or $\goal$, and in particular guarantees the existence of a concrete strategy (as defined in~\Cref{sec:prelims}) conforming to~$\sstrat$.

\begin{lemma}
  If \playerR{} (\playerS{}) has a winning symbolic strategy in $s$, then \playerR{} (\playerS{}) has a concrete winning strategy in $s$.
\end{lemma}
\begin{proof}
  Let $\sstrat$ by a symbolic winning strategy for \playerR{}.
  Let $\sigma_R$ be any reachability strategy such that for all play prefixes $\omega s \in S^* S_{\playerR}$ that conform to $\sstrat$ the formula $\sstrat(s,\sigma_R(\omega s))$ is valid. 
  Such a function is guaranteed to exist, as $(\sstrat \land \reach)(s)$ is satisfiable for all such play prefixes by definition.
  Furthermore, $\sigma_R$ is winning as all play prefixes of plays consistent with $\sigma_R$ conform to $\sstrat$, and hence all these plays are winning by assumption. 
  The proof for \playerS{} is analogous.
  \qed
\end{proof}

This representation allows us to specify nondeterministic strategies, but classical memoryless strategies on finite arenas (specified as a function $\sigma\colon S_{\playerR}\to S$ or $\S_{\playerS}\to S$) can also be represented in this form using a disjunction over formulas $\bigwedge_{v\in \var} (v =s(v)\land v' =\sigma(s)(v))$ for varying $s\in S$.

The following lemma shows that a necessary subgoal directly yields a symbolic strategy for \playerS{} if the subgoal is, in a certain sense, \emph{locally avoidable} by \playerS{}. It will be our main tool for synthesizing safety player strategies.
\begin{restatable}{lemma}{avoidingC}
	\label{lem:avoidingC}
	Let $C$ be a necessary subgoal for $\G$ and suppose that $\unsat(\enf(C,\G))$ holds.
	Then, $\safe \land \neg C$ is a winning symbolic strategy for $\playerS$ in $\G$.
\end{restatable}

\subsection{A Recursive Algorithm}

\begin{algorithm}
	\caption{$\algreach(\G)$}
	\label{alg:algreach}
	\SetKwInOut{Input}{In}
	\SetKwInOut{Output}{Out}
	\Input{ reachability game \game{}}
	\Output{ triple $(R, \sstrat_{\playerR}, \sstrat_{\playerS})$ s.t.
		\begin{itemize}
			\item $\mathit{R}\in\logic(\var)$ represents the set of initial states winning for \playerR{};
			\item $\sstrat_{\playerR}$ is a winning symbolic reachability strategy  for states in $R$;
			\item $\sstrat_{\playerS}$ is a winning symbolic safety strategy for states in $\init \land \neg R$.
		\end{itemize}}
	\Begin{
		$\mathit{R} \leftarrow \init \land \goal$\\
		$\mathit{I} \leftarrow \init \land \lnot\goal$\\
		\If{$\unsat(\mathit{I})$}{\KwRet{$\mathit{R}, {\normalfont \false}, \normalfont \false$ \label{line:ret1}}}
		$\varphi \leftarrow  \interpol(\mathit{I},\goal)$\label{line: phi}\\
		$\mathit{C} \leftarrow \instantiate(\varphi,\G)$\label{line: subgoal}\\
		\If{$\unsat(\enf(\mathit{C},\G))$\label{line: cpre}}{\KwRet{$\mathit{R}, {\normalfont \false},\safe \land \neg C$}\label{line:safety wins}}
		$\G_{post} \leftarrow \langle \post(C)\subsm, \safe \land \varphi, \reach \land \varphi, \goal \rangle$\label{line:gpost}\\
		$\mathit{R}_{post},\sstrat_{\playerR{}}^{post}, \sstrat_{\playerS{}}^{post} \leftarrow \algreach \big(\G_{post}\big)$\label{line: recursion1}\\
		$\mathit{F} \leftarrow C \land \mathit{R}_{post}\subsp$\\
		\If{$\unsat(\enf(\mathit{F},\G))$\label{line:safeavoidsE}}{\KwRet{$\mathit{R}, {\normalfont \false}, \safe \land \neg F \land (\varphi \implies \sstrat_{\playerS{}}^{post})$}\label{line:thirdreturn}}
		\If{$\sat((\reach \lor \safe) \land \varphi \land \lnot \varphi' \land \lnot \goal)$\label{line:transback}}{
			$\mathit{F} \leftarrow \mathit{F} \lor \goal \subsp$\label{line: E2}\\
			$\varphi \leftarrow \false$\label{line:tpostfalse}
		}
		$\G_{pre} \leftarrow \langle I,\safe \land \lnot \mathit{F},\reach \land \lnot \mathit{F}, \pre(\enf(\mathit{F},\G))\rangle$\label{line:gpre2}\\
		$\mathit{R}_{pre}, \sstrat_{\playerR{}}^{pre}, \sstrat_{\playerS{}}^{pre} \leftarrow \algreach \big(\G_{pre}\big)$\label{line: recursion2}\\
    \KwRet{$\mathit{R} \lor \mathit{R}_{pre}$, \label{line: last return}\\ $\qquad \quad \;\; \combine(\sstrat_{\playerR{}}^{pre},F, \sstrat_{\playerR{}}^{post})$\label{line:reachstrat}, \\$\qquad \quad \;\;(\neg \varphi \implies \sstrat_{\playerS{}}^{pre}) \land (\varphi \implies \sstrat_{\playerS{}}^{post})$\label{line:safetystrat}}
	}
\end{algorithm}

We now describe our algorithm which utilizes necessary subgoals to decompose and solve two-player reachability games (Algorithm~\ref{alg:algreach}).
It is incomplete in the sense that it does not return on every input (\Cref{sec:termination} discusses special cases with guaranteed termination).
If the algorithm returns on input $\G$, it returns a triple $(R,\sstrat_{\playerR{}},\sstrat_{\playerS{}})$, where (1) $R$ is a state predicate characterizing the initial states that are winning for \playerR{} in $\G$, (2) $\sstrat_{\playerR{}}$ is a symbolic strategy for \playerR{} that wins in all initial states satisfying $R$, and (3) $\sstrat_{\playerS{}}$ is a symbolic strategy for \playerS{} that wins in all initial states satisfying $\init \land \neg R$.
The returned safety strategy $\sstrat_{\playerS{}}$ is such that $\neg \sstrat_{\playerS{}}$ is a necessary subgoal that \playerS{} can avoid locally in the game $\G$ restricted to intial states $\init \land \neg R$ (see~\Cref{lem:avoidingC}).

Algorithm~\ref{alg:algreach} works as follows.
States satisfying $\init$ and $\goal$ are immediately winning for $\playerR$ and thus always part of the returned formula $R$.
Following the discussion at the beginning of Section~\ref{sec:subgoals}, further analysis considers the game starting in the remaining initial states $I = \init \land \lnot \goal$.
If there is no such state, we may return that all initial states are winning (line~\ref{line:ret1}).
Here, \playerR{} wins from $R$ without playing any move, and hence $\sstrat_{\playerR{}} = \false$ is a valid winning symbolic strategy (winning symbolic strategies are only required to provide moves in prefixes that have not seen $\goal$ so far).
We may choose $\sstrat_{\playerS{}}$ arbitrarily as there is no initial state winning for $\playerS{}$.

If the algorithm does not return in line~\ref{line:ret1}, a necessary subgoal $C$ between $I$ and $\goal$ is computed by instantiating a Craig interpolant $\varphi$ for the two predicates (lines \ref{line: phi} and \ref{line: subgoal}, see also \Cref{prop_necessary}).
We break up the remaining description of the algorithm into three parts, which correspond to the main cases that occur when splitting the game along the subgoal $C$.

\subsubsection{Case 1: \playerS{} can avoid the subgoal $C$.}
If the necessary subgoal $C$ qualifies for~\Cref{lem:avoidingC}, we can immediately conclude that \playerS{} is winning for all states statisfying $I$ (lines~\ref{line: cpre} and~\ref{line:safety wins}).
An instance of this case occurs if the interpolant describes a bottleneck in the game which is fully controlled by \playerS{}.
The winning symbolic reachability strategy is $\safe \land \neg C$ in this case (line~\ref{line:safety wins}), and we will assume that safety strategies returned by recursive calls of the algorithm are essentially negations of necessary subgoals that can be avoided by \playerS{}. 

If Lemma \ref{lem:avoidingC} is not applicable, we next find those transitions in $C$ that move into a winning state for the safety player.
This is achieved by analyzing the \emph{post-game} (line~\ref{line:gpost}):
\[\G_{post} = \langle \post(C)\subsm, \safe \land \varphi, \reach \land \varphi, \goal \rangle.\]

Its initial states are exactly the states one sees after bridging the subgoal $C$.
In order to make sure that $\G_{post}$ is, in some sense, easier to solve than $\G$, we restrict both $\safe$ and $\reach$ to $\varphi$, which is the interpolant used to compute the subgoal $C$.
This has the effect of removing all transitions in states \emph{not} satisfying~$\varphi$, making them trap states.
For the safety player this makes $\G_{post}$ easier to win than $\G$ as all plays ending in such a trap state without seeing $\goal$ before are winning for \playerS{} in $\G_{post}$.
Hence we formally have:
\begin{restatable}{lemma}{gpost}
  \label{lem:gpost}
  If $\sstrat$ is a winning symbolic reachability strategy from $s$ in $\G_{post}$, then $\sstrat$ is also winning from $s$ in $\G$.
\end{restatable}

Due to the restriction to $\varphi$, intuitively \playerR{} wins from a state $s$ in $\G_{post}$ if they can win from $s$ in $\G$ \emph{while staying inside the interpolant} $\varphi$.
In other words, \playerR{} must guarantee that the necessary subgoal $C$ is not visited again in the play.
Still, the set $R_{post}$, as returned in line \ref{line: recursion1} by the recursive call to Algorithm~\ref{alg:algreach} on $\G_{post}$, is a sufficient subgoal in $\G$, by the above lemma.
Furthermore, if \playerS{} can avoid all states satisfying $R_{post}$ (see line~\ref{line:safeavoidsE}), then this also implies a winning strategy from all initial states in $I$. 
The reason is that \playerR{} can only win by eventually visiting a state from which they can win without leaving $\varphi$ again, as $(\goal \implies \varphi)$ is valid.
This is not possible if \playerS{} can avoid all states in $R_{post}$. 

In this case we construct $\sstrat_{\playerS{}}$ as follows.
We assume that $\neg \sstrat_{\playerS{}}^{post}$ is a necessary subgoal that can be locally avoided in $\G_{post}$ from all states satisfying $\post(C)\subsm \land \neg R_{post}$, and furthermore, we know that $F := C \land R_{post}\subsp$ can be locally avoided in $\G$ (line~\ref{line:safeavoidsE}).
Intuitively, playing according to $\sstrat^{post}_{\playerS{}}$ in $\G_{post}$ yields a strategy for \playerS{} which avoids $\goal$ and may move back into a state satisfying $\neg \varphi$, which forces \playerR{} to bridge the subgoal $C$ again in order to win.
It follows that $F \lor (\varphi \land \neg \sstrat_{\playerS{}}^{post})$ is a necessary subgoal from $I$ that can be locally avoided by \playerS{} in $\G$, and the corresponding symbolic strategy is $\safe \land \neg F \land (\varphi \implies \sstrat_{\playerS{}}^{post})$ (we additionally intersect the negated necessary subgoal with $\safe$ to ensure that the symbolic strategy only includes legal transitions).

\medskip
So far, the subgoal was such that \playerS{} could avoid it entirely, or at least avoid all states from which \playerR{} would win when forced to remain inside the post-game.
If this is not the case, then we also need to consider the \emph{pre-game} (line~\ref{line:gpre2}):
\[\G_{pre} = \langle I,\safe \land \lnot \mathit{F},\reach \land \lnot \mathit{F}, \pre(\enf(\mathit{F},\G))\rangle.\]
which intuitively describes the game before bridging the interpolant $C$ for the last time.
The exact definition of $F$ will depend on whether $C$ \emph{perfectly partitions} the game or not.
In both cases $F$ will be the largest sufficient subgoal contained in a necessary subgoal, which lets us apply~\Cref{lem:slicing} to conclude that the initial winning regions of $\G$ and $\G_{pre}$ coincide.

\subsubsection{Case 2: The subgoal perfectly partitions $\G$.}
We say that $\varphi$ \emph{perfectly partitions} $\G$ if $(\reach \lor \safe) \land \varphi \land \neg \varphi' \land \neg \goal$ is unsatisfiable (cf. line~\ref{line:transback}).
Intuitively, this means that there is no transition that ``undoes'' the effect of the subgoal~$C$.
If this holds, then the restriction of $\G_{post}$ to states satisfying $\varphi$ is \emph{de facto} no longer a restriction, as no play can reach such a state anyway after passing through the subgoal.
This intuition is formalized by the following lemma.
\begin{restatable}{lemma}{perfectpartition}
  Assume that $\varphi$ \emph{perfectly partitions} $\G$, and let $s$ be a state satisfying $\post(C)\subsm$. Then \playerR{} wins from $s$ in $\G_{post}$ if and only if \playerR{} wins from $s$ in~$\G$.
\end{restatable}
It follows that $F = C \land R_{post} \subsp$ is the largest sufficient subgoal included in~$C$.
By~\Cref{lem:slicing}, the same initial states are winning for \playerR{} in $\G_{pre}$ and in~$\G$. 
In this case, we construct the desired safety strategy (line~\ref{line:safetystrat}) as 
\[\sstrat_{\playerS{}} = (\neg \varphi \implies \sstrat_{\playerS{}}^{pre}) \land (\varphi \implies \sstrat_{\playerS{}}^{post}),\]
where $\neg \sstrat_{\playerS{}}^{pre/post}$ are assumed to be necessary subgoals avoidable by \playerS{} in the corresponding subgames. 
Intuitively, the combined strategy consists of following $\sstrat_{\playerS}^{pre}$ as long as one remains in the pre-game, which, by induction hypothesis, allows \playerS{} to avoid all transitions from $F$ if starting in $R_{pre}$. 
If the play crosses $C \land \neg F$, the strategy is to play according to the winning strategy of the post-game.

A symbolic strategy for \playerR{} can be given by combining pre- and post-strategies as follows (line~\ref{line:reachstrat}):
\begin{align*}
  \combine(\sstrat_{\playerR{}}^{pre},F, \sstrat_{\playerR{}}^{post}) := \;\;&(\pre(\sstrat_{\playerR{}}^{post}) \implies \sstrat_{\playerR{}}^{post}) \\
  \land \;\; &((\neg \pre(\sstrat_{\playerR{}}^{post}) \land \pre(F)) \implies F) \\
  \land \;\; &((\neg \pre(\sstrat_{\playerR{}}^{post}) \land \neg \pre(F)) \implies \sstrat_{\playerR{}}^{pre}) \\
  \land \;\; &(\pre(\sstrat_{\playerR{}}^{post}) \lor \pre(F) \lor \pre(\sstrat_{\playerR{}}^{pre})).
\end{align*}
This represents a nested conditional strategy that prefers the strategies of the subgames in the priority order $\sstrat_{\playerR{}}^{post}, F$, and finally $\sstrat_{\playerR{}}^{pre}$.
The reason for this order is that the winning condition in the post-game coincides with the global winning objective (to reach $\goal$), while in the pre-game \playerR{} tries to reach a winning state in the post-game.
The set $F$ is exactly the bridge between these two.
The last condition makes sure that the strategy only includes transitions of states in which it is winning.

\subsubsection{Case 3: The subgoal does not perfectly partition $\G$.}

If $\sat((\reach \lor \safe) \land \varphi \land \lnot \varphi' \land \neg \goal)$ is true in line~\ref{line:transback}, we can no longer assume that $F$ is the largest sufficient subgoal in $C$.
The reason is that \playerS{} may win in $\G_{post}$ by moving out of the subgame, but if this move leads to a winning state for \playerR{} in $\G$, then such a strategy is winning in $\G_{post}$, but not in $\G$.
So we can only assume that $F$ is sufficient (this follows by~\Cref{lem:gpost}).
In order to apply~\Cref{lem:slicing} we extend~$F$ by all transitions that move directly into $\goal$ (line~\ref{line: E2}). 
This immediately yields a necessary and sufficient subgoal, and so again~\Cref{lem:slicing} applies to $\G_{pre}$ (line~\ref{line:gpre2}).
We could have also added $\goal$-states to $F$ in Case 2, but we have observed that not doing so improves the performance of our procedure considerably. 

The reachability strategy is composed of $\sstrat_{\playerR{}}^{pre}, F$, and $\sstrat_{\playerR{}}^{post}$ exactly as in Case 2 (line~\ref{line:reachstrat}).
As all transitions in $F$ are losing for \playerS{}, and these are the only ones that are removed in $\G_{pre}$, essentially \playerS{} can play using the same strategies in $\G$ and $\G_{pre}$.
We implement this by setting $\varphi$ to $\false$ (line~\ref{line:tpostfalse}), in which case $\sstrat_{\playerS}$ (line~\ref{line:safetystrat}) equals 
$(\true \implies \sstrat_{\playerS{}}^{pre}) \land (\false \implies \sstrat_{\playerS}^{post}) \equiv \sstrat_{\playerS{}}^{pre}$.

\medskip
Finally, we formally state the partial correctness of the algorithm, using the ideas outlined above. 
The proof can be found in the appendix.

\begin{restatable}[Partial correctness]{theorem}{partialcorrectness}
  \label{thm:partcorr}
  If $\algreach(\G)$ returns $(R,\sstrat_{\playerR{}},\sstrat_{\playerS{}})$, then
  \begin{itemize}
  \item $R$ characterizes the set of initial states that are winning for $\playerR{}$ in $\G$,
  \item $\sstrat_{\playerR{}}$ is a winning symbolic reachability strategy from $R$,
  \item $\sstrat_{\playerS{}}$ is a winning symbolic safety strategy from $\init \land \neg R$.
  \end{itemize}
\end{restatable}

\begin{remark}[Simulating the attractor]
  \label{rem:attractor}
	Note that Craig interpolants are by no means unique. If we choose the interpolation function so that $\interpolate(I,\goal)$ always returns $\goal$ (this is a valid interpolant), Algorithm~\ref{alg:algreach} essentially simulates the attractor.
  In this case the subgoal $C$ (line~\ref{line: subgoal}) contains exactly the transitions that move directly into $\goal$.
  All states in $\post(C)\subsm$ are then winning for \playerR{} and hence $R_{post}$ would be equivalent to $\post(C)\subsm$, which implies that $C \equiv F$ holds in this case.
  The new goal states in $\G_{pre}$ are set to $\pre(\enf(F,\G))$, which are exactly the states in $\pre(C)$ that either are controlled by \playerR{}, or such that all their transitions are included in $F$.
  Hence the set $\pre(\enf(F,\G))$ is exactly the classical controlled predecessor.
\end{remark}

One effect of slicing the game along general subgoals is that the initial predicate of the post-game (which describes all states satisfying the post-condition of the subgoal) may be satisfied by many states that do not necessarily need to be considered in order to decide who wins from the initial states of $\G$ (for example, because they are not reachable from any initial state, or cannot reach $\goal$).
This can be a drawback if the (superfluous) size of the subgames makes them hard to solve.
Notably, this is in general less of an issue for approaches based on unrolling of the transition relation: 
The method of solving increasingly large step-bounded games~\cite{FarzanK17} will only consider states that are reachable from $\init$, while backwards fixpoint computations will not explore states that do not reach $\goal$.
 A way of coping with this is to provide additional information on the domains of variables, whenever this is available (we discuss the effect of bounding variable domains in \Cref{sec:experiments}).
Indeed, in the case where all variable domains are finite, Algorithm~\ref{alg:algreach} is guaranteed to terminate, as shown in the next subsection.

\subsection{Special Cases with Guaranteed Termination}
\label{sec:termination}

Deciding the winner in the types of games we consider is generally undecidable (see~\cite{FarzanK17} for the case that $\logic$ is linear real arithmetic).
Since Algorithm \ref{alg:algreach} returns a correct result whenever it terminates, this implies that it cannot always terminate.
In this section, we give two important cases in which we can prove termination. 
The proofs can be found in the appendix.

\begin{restatable}{theorem}{terminationfinite}
  If the domains of all variables in $\G$ are finite, then $\algreach(\G)$ terminates.
\end{restatable}

\begin{remark}[Time complexity]
  The termination argument given in the appendix yields a single-exponential upper bound on the runtime of the algorithm, where the input size is measured in the number of concrete transitions of the game.
  This is because in both recursive calls the subgames may be ``almost'' as large as the input -- they are only guaranteed to have at least one concrete transition less.
\end{remark}

We now show that, under certain assumptions, our algorithm also terminates for games that have a finite bisimulation quotient.
To this end, we first clarify what bisimilarity means in our setting.
A relation $R \subseteq S \times S$ over the states of $\G$ is called a \emph{bisimulation} on $\G$, if
\begin{itemize}
	\item for all $(s_1,s_2) \in R$ the formulas $\goal(s_1) \iff \goal(s_2)$, $\init(s_1) \iff \init(s_2)$ and $\reachvar(s_1) \iff \reachvar(s_2)$ are valid (recall that $\reachvar$ holds exactly in states controlled by \playerR{}).
\item for all $(s_1,s_2) \in R$ and $s_1' \in S$ such that $(\safe \lor \reach)(s_1,s_1')$ holds, there exists $s_2' \in S$ such that $(\safe \lor \reach)(s_2,s_2')$ holds, and $(s_1',s_2') \in R$.
\item for all $(s_1,s_2) \in R$ and $s_2' \in S$ such that $(\safe \lor \reach)(s_2,s_2')$ holds, there exists $s_1' \in S$ such that $(\safe \lor \reach)(s_1,s_1')$ holds, and $(s_1',s_2') \in R$.
\end{itemize}
We say that $s_1$ and $s_2$ are \emph{bisimilar} (denoted by $s_1 \sim s_2$) if there exists a bisimulation $R$ such that $(s_1,s_2) \in R$.
Bisimilarity is an equivalence relation, and it is the coarsest bisimulation on $\G$.
The equivalence classes are called \emph{bisimulation classes}.
As the winning region of any player can be expressed in the $\mu$-calculus~\cite{Zappe01} and the $\mu$-calculus is invariant under bisimulation~\cite{BradfieldS07}, it follows that bisimilar states are won by the same player.
\begin{lemma}
  \label{lem:bisimpreservesreach}
  Let $R$ be a bisimulation on $\G$.
  If $(s_1, s_2) \in R$, then \playerR{} wins from $s_1$ in $\G$ if and only if \playerR{} wins from $s_2$ in $\G$.
\end{lemma}

We will assume that for each bisimulation class $S_i$ there exists a formula $\psi_i \in \logic(\var)$ that \emph{defines} $S_i$, formally: 
For all $s \in S$, $\psi_i(s)$ holds if and only if $s \in S_i$.
Furthermore, we will assume that the interpolation procedure \emph{respects}~$\sim$, formally: $\interpolate(\varphi,\psi)$ is equivalent to a disjunction of formulas $\psi_i$.
Such an interpolant exists if $\psi$ or $\varphi$ already satisfy this assumption.

\begin{restatable}{theorem}{terminationbisim}
  Let $\G$ be a reachability game with finite bisimulation quotient under $\sim$ and assume that all bisimulation classes of $\G$ are definable in $\logic$.
  Furthermore, assume that $\interpolate$ respects $\sim$.
  Then, $\algreach(\G)$ terminates.
\end{restatable}

\section{Case Studies}
\label{sec:experiments}

In this section we evaluate our approach on a number of case studies. 
Our prototype \cabpy{}\footnote[2]{The source code of \cabpy{} and our experimental data are both available at \url{https://github.com/reactive-systems/cabpy}. We provide a virtual machine image with \cabpy{} already installed for reproducing our evaluation~\cite{VM}.} is written in Python and implements the game solving part of the presented algorithm. 
Extending it to returning a symbolic strategy using the ideas outlined above is straightforward.
We compared our prototype with \simsynth \cite{FarzanK17}, the only other readily available tool for solving linear arithmetic games. 
The evaluation was carried out with Ubuntu 20.04,  a 4-core Intel\textsuperscript{\textregistered} Core\texttrademark~i5 2.30GHz processor, as well as 8GB of memory. \cabpy{} uses the PySMT~\cite{pysmt2015} library as an interface to the MathSAT5~\cite{mathsat5} and Z3~\cite{z3solver} SMT solvers.
On all benchmarks, the timeout was set to 10 minutes. In addition to the winner, we report the runtime and the number of subgames our algorithm visits. Both may vary with different SMT solvers or in different environments.
\subsection{Game of Nim}

Game of Nim is a classic game from the literature \cite{Bouton1901} and played on a number of heaps of stones. Both players take turns of choosing a single heap and removing at least one stone from it. We consider the version where the player that removes the last stone wins. Our results are shown in \Cref{exp_nim}. In instances with three heaps or more we bounded the domains of the variables in the instance description, by specifying that no heap exceeds its initial size and does not go below zero.

Following the discussion in \Cref{sec:termination}, we need to bound the domains to ensure the termination of our tool on these instances.
Remarkably, bounding the variables was not necessary for instances with only two heaps, where our tool \cabpy{} scales to considerably larger instances than \simsynth.
We did not add the same constraints to the input of \simsynth{}, as for \simsynth{} this resulted in longer runtimes rather than shorter.
In Game of Nim, there are no natural necessary subgoals that the safety player can locally control.

The results (see~\Cref{exp_nim}) demonstrate that our approach is not completely dependent on finding the right interpolants and is in particular also competitive when the reachability player wins the game. We suspect that \simsynth{} performs worse in these cases because the safety player has a large range of possible moves in most states, and inferring the win of the reachability player requires the tool to backtrack and try our all of them.

\begin{figure}[tbp]
	\centering
	{\def\arraystretch{1.1}\tabcolsep=5pt
    \small
		\begin{tabular}{|c|c|c|c|c|}
			\hline
			& \multicolumn{2}{c|}{\cabpy} & \simsynth & \\
			\hline
			Heaps & Subgames & Time(s) &Time(s) & Winner\\
			\hline
			(4,4) & 19 & 1.50 & 10.44 & $\playerR$\\
			(4,5) & 23 & 1.92 & 12.74 & $\playerS$\\
			(5,5) & 23 & 1.99 & 85.75 & $\playerR$\\
			(5,6) & 27 & 2.90 & 91.66 & $\playerS$\\
			(6,6) & 28 & 3.04 & Timeout & $\playerR$\\
			(6,7) & 31 & 3.76 & Timeout & $\playerS$\\
			(20,20) & 88 & 94.85 & Timeout & $\playerR$\\
			(20,21) & 94 & 113.04 & Timeout & $\playerS$\\
			(30,30) & 128 & 364.13 & Timeout & $\playerR$\\
			(30,31) & 135 & 404.02 & Timeout & $\playerS$\\\hline
			(3,3,3)b & 23 & 13.63 & 2.85 & $\playerS$\\
			(1,4,5)b & 32 & 7.00 & 289.85 & $\playerR$\\
			(4,4,4)b & 33 & 50.55 & 24.39 & $\playerS$\\
			(2,4,6)b & 38 & 19.77 & Timeout & $\playerR$\\
			(5,5,5)b & 33 & 127.89 & 162.50 & $\playerS$\\
			(3,5,6)b & 40 & 86.56 & Timeout & $\playerR$\\\hline
			(2,2,2,2)b & 39 & 84.79 & 213.79 & $\playerR$\\
			(2,2,2,3)b & 41 & 102.01 & Timeout & $\playerS$\\
			\hline
	\end{tabular}}
	\caption{Experimental results for the Game of Nim. The notation $(h_1,\ldots, h_n)$ denotes the instance played on $n$ heaps, each of which consists of $h_i$ stones. Instances marked with b indicate that the variable domains were explicitly bounded in the input for \cabpy{}.}
	\label{exp_nim}
\end{figure}

\begin{figure}[tbp]
	\centering
	{\def\arraystretch{1.1}\tabcolsep=5pt
    \small
		\begin{tabular}{|c|c|c|c|c|}
			\hline
			& \multicolumn{2}{c|}{\cabpy} & \simsynth & \\
			\hline
			$r$ & Subgames & Time(s) &Time(s) & Winner\\
			\hline
			10 & 10 & 0.57 & 3.93 & $\playerS$\\
			20 & 20 & 1.23 & 20.48 & $\playerS$\\
			40 & 40 & 3.42 & 121.96 & $\playerS$\\
			60 & 60 & 7.36 & Timeout & $\playerS$\\
			80 & 80 & 17.72 & Timeout & $\playerS$\\
			100 & 100 & 26.36 & Timeout & $\playerS$\\
			\hline
	\end{tabular}}
	\caption{Experimental results for the Corridor game. The safety player controls the door between rooms $r-1$ and $r$.}
	\label{exp_corridor}
\end{figure}

\subsection{Corridor}

We now consider an example that demonstrates the potential of our method in case the game structure contains natural bottlenecks. Consider a corridor of $100$ rooms arranged in sequence, i.e., each room $i$ with $0 \leq i < 100$ is connected to room $i+1$ with a door. The objective of the reachability player is to reach room 100 and they are free to choose valid values from $\mathbb{R}^2$ for the position in each room at every other turn. The safety player controls some door to a room $r \leq 100$. Naturally, a winning strategy is to prevent the reachability player from passing that door, which is a natural bottleneck and necessary subgoal on the way to the last room.

 The experimental results are summarized in Figure \ref{exp_corridor}. We evaluated several versions of this game, increasing the length from the start to the controlled door. The results confirm that our causal synthesis algorithm finds the trivial strategy of closing the door quickly. This is because Craig interpolation focuses the subgoals on the room number variable while ignoring the movement in the rooms in between, as can be seen by the number of considered subgames. \simsynth, which tries to generalize a strategy obtained from a step-bounded game, struggles because the tool solves the games that happen between each of the doors before reaching the controlled one.

\subsection{Mona Lisa}

The game described in Section \ref{sec:motivation} between a thief and a security guard is very well suited to further assess the strength and limitations of both our approach as well as of \simsynth{}. We ran several experiments with this scenario, scaling the size of the room and the sleep time of the guard, as well as trying a scenario where the guard does not sleep at all. Scaling the size of the room makes it harder for \simsynth{} to solve this game with a forward unrolling approach, while our approach extracts the necessary subgoals irrespective of the room size. However, scaling the guard's sleep time makes it harder to solve the subgame between the two necessary subgoals, while it only has a minor effect on the length of the unrolling needed to stabilize the play in a safe region, as done by \simsynth.

 The results in Figure \ref{exp_monalisa} support this conjecture. The size of the room has \emph{almost no effect at all} on both the runtime of \cabpy{} and the number of considered subgames. However, as the results for a sleep value of 4 show, the employed combination of quantifier elimination and interpolation introduces some instability in the produced formulas. This means we may get different Craig interpolants and slice the game with more or less subgoals. Therefore, we see a lot of potential in optimizing the interplay between the employed tools for quantifier elimination and interpolation. The phenomenon of the runtime being sensitive to these small changes in values is also seen with \simsynth, where a longer sleep time sometimes means a faster execution.

\begin{figure}[tbp]
	\centering
    {\def\arraystretch{1.1}\tabcolsep=5pt
      \small
		\begin{tabular}{|c|c|c|c|c|c|}
			\hline
			&  & \multicolumn{2}{c|}{\cabpy} & \simsynth & \\
			\hline
			Size & Sleep & Subgames & Time(s) & Time(s) & Winner\\
			\hline
			$10\times10$ & - & 7 & 0.61 & 4.79 & $\playerS$\\
			$20\times20$ & - & 7 & 0.60 & 25.26 & $\playerS$\\
			$40\times40$ & - & 7 & 0.61 & 157.62 & $\playerS$\\
			\hline
			$10\times10$ & 1 & 10 & 4.22 & 20.31 & $\playerS$\\
			$20\times20$ & 1 & 11 & 4.34 & 36.44 & $\playerS$\\
			$40\times40$ & 1 & 11 & 4.65 & 226.14 & $\playerS$\\
			\hline
			$10\times10$ & 2 & 13 & 5.88 & 7.40 & $\playerS$\\
			$20\times20$ & 2 & 14 & 5.98 & 60.00 & $\playerS$\\
			$40\times40$ & 2 & 13 & 5.92 & 270.48 & $\playerS$\\
			\hline
			$10\times10$ & 3 & 18 & 26.58 & 13.94 & $\playerS$\\
			$20\times20$ & 3 & 17 & 26.19 & 115.53 & $\playerS$\\
			$40\times40$ & 3 & 18 & 27.85 & 290.12 & $\playerS$\\
			\hline
			$10\times10$ & 4 & 30 & 175.27 & 13.96 & $\playerS$\\
			$20\times20$ & 4 & 22 & 204.79 & 60.08 & $\playerS$\\
			$40\times40$ & 4 & 27 & 123.95 & 319.47 & $\playerS$\\
			\hline
	\end{tabular}}
	\caption{Experimental results for the Mona Lisa game.}
	\label{exp_monalisa}
\end{figure}

\subsection{Program Synthesis}

Lastly, we study two benchmarks that are directly related to program synthesis. 
The first problem is to synthesize a controller for a thermostat by filling out an incomplete program, as described in \cite{BeyeneCPR14}. A range of possible initial values of the room temperature $c$ is given, e.g., $20.8 \leq c \leq 23.5$, together with the temperature dynamics which depend on whether the heater is on (variable $o \in \mathbb{B}$).
The objective for $\playerS$ is to control the value of $o$ in every round such that $c$ stays between $20$ and $25$. This is a common benchmark for program synthesis tools and both \cabpy{} and \simsynth{} solve it quickly.
The other problem relates to Lamport's bakery algorithm\cite{Lamport1974}. We consider two processes using this protocol to ensure mutually exclusive access to a shared resource. The game describes the task of synthesizing a scheduler that violates the mutual exclusion. This essentially is a model checking problem, and we study it to see how well the tools can infer a safety invariant that is out of control of the safety player. For our approach, this makes no difference, as both players may play through a subgoal and the framework is well suited to find a safety invariant. The forward unrolling approach of \simsynth, however, seems to explore the whole state space before inferring safety, and fails to find an invariant before a timeout. 

\begin{figure}[tbp]
	\centering
	{\def\arraystretch{1.1}\tabcolsep=5pt
    \small
		\begin{tabular}{|c|c|c|c|c|}
			\hline
			& \multicolumn{2}{c|}{\cabpy} & \simsynth & \\
			\hline
			Name & Subgames & Time(s) &Time(s) & Winner\\
			\hline
			Thermostat & 6 & 0.44 & 0.39 & $\playerS$ \\
			Bakery & 46 & 18.25 & Timeout & $\playerS$ \\
			\hline
	\end{tabular}}
	\caption{Experimental results for program synthesis problems.}
	\label{exp_synth}
\end{figure}

\section{Conclusion}
Our work is a step towards the fully automated synthesis of software. 
It targets symbolically represented reachability games which are expressive enough to model a variety of problems, from common game benchmarks to program synthesis problems. 
The presented approach exploits causal information in the form of \emph{subgoals}, which are parts of the game that the reachability player needs to pass through in order to win.
Having computed a subgoal, which can be done using Craig interpolation, the game is split along the subgoal and solved recursively.
At the same time, the algorithm infers a structured symbolic strategy for the winning player.
The evaluation of our prototype implementation \textsc{CabPy} shows that our approach is practically applicable and scales much better than previously available tools on several benchmarks.  
While termination is only guaranteed for games with finite bisimulation quotient, the experiments demonstrate that several infinite games can be solved as well.

This work opens up several interesting questions for further research.
One concerns the quality of the returned strategies.
Due to its compositional nature, at first sight it seems that our approach is not well-suited to handle global optimization criteria, such as reaching the goal in fewest possible steps.
On the other hand, the returned strategies often involve only a few key decisions and we believe that therefore the strategies are often very sparse, although this has to be further investigated.
We also plan to automatically extract deterministic strategies from the symbolic ones~\cite{Bloem,Ehlers} we currently consider.

Another question regards the computation of subgoals. 
The performance of our algorithm is highly influenced by which interpolant is returned by the solver.
In particular this affects the number of subgames that have to be solved, and how complex they are.
We believe that template-based interpolation \cite{template_interpolation} could be a promising candidate to explore with the goal to compute good interpolants. 
This could be combined with the possibility for the user to provide templates or expressive interpolants directly, thereby benefiting from the user's domain knowledge.

\bibliographystyle{splncs04}
\bibliography{lit}
\section*{Appendix}

\avoidingC*
\begin{proof}
	We first show that all plays that conform to the symbolic safety strategy $\safe \land \neg C$ are winning for \playerS{}.
	Let $\rho = s_0s_1 \ldots$ be such a play and assume, for contradiction, that $n \in \mathbb{N}$ exist such that $\goal(s_n)$ is valid, and let $n$ be the least such $n$.
	From $\unsat(\enf(C,\G))$ it follows that all states in $\pre(C)$ are under the control of \playerS{}, and hence $C(s_j,s_{j+1})$ can only be true if $s_j \in S_{\playerS}$.
	But as $\rho$ conforms to $\safe \land \neg C$, it follows that all play prefixes of $\rho$ conform to $\safe \land \neg C$, in particular, for all $j < n$ the formula $\neg C(s_j,s_{j+1})$ must hold.
  This is in contradiction to the fact that $C$ is a necessary subgoal.
	
	Now we show that if $\pi = s_0 \ldots s_n \in S^*S_{\playerS}$ is a play prefix that conforms to $\safe \land \neg C$, then $\pre(\safe \land \neg C)(s_n)$ holds.
	As $\pi$ is a play prefix and $s_n \in S_{\playerS}$, $\safe(s_n)$ is satisfiable.
	From $\unsat(\enf(C,\G))$ it follows that if $(\safe \land C)(s_j)$ is satisfiable, then so is $(\safe \land \neg C)(s_j)$, which concludes the proof.
	\qed
\end{proof}

\gpost*
\begin{proof}
	First, we observe that $\sstrat \implies ((\safe \lor \reach) \land \varphi)$ is valid by assumption, and hence so is $\sstrat \implies (\safe \lor \reach)$.
	We show that all plays of $\G$ that conform to $\sstrat$ are winning for \playerR{}.
	Let $\rho = s_0 s_1 \ldots$ be such a play, and assume for contradiction that it is winning for \playerS{}.
	As playing according to $\sstrat$ is winning for \playerR{} in $\G_{post}$, the play needs to see a transition that is not included in $\G_{post}$.
	Such a transition can only be played by \playerS{}, as \playerR{} plays according to $\sstrat$, and $\sstrat \implies ((\safe \lor \reach) \land \varphi)$ is valid.
	This can only happen if there exists $i \in \mathbb{N}$ such that $(\safe\land \lnot\varphi)(s_i, s_{i+1})$ holds, or in other words, $\safe(s_i, s_{i+1})$ and $\lnot\varphi(s_i)$. Let $i$ be the first such index. Then $s_i$ does not have any outgoing transitions in $\G_{post}$, meaning that $s_0\ldots s_i$ is a finite play in $\G_{post}$, and hence there must exist $0\leq j\leq i$ such that $\goal(s_j)$ is satisfied.
	Otherwise, the play would be winning for \playerS{} in $\G_{post}$, contradicting the assumption.
	Hence, $\rho$ is winning for \playerR{}.
	
	It remains to show that any play prefix $s_0 \ldots s_n \in S^*S_{\playerR}$ that conforms to $\sstrat$ is such that $(\sstrat \land \reach)(s_n)$ is satisfiable.
  If the play prefix remains inside $\G_{post}$, then this follows from the fact that for all such prefixes $(\sstrat \land \reach(s_n) \land \varphi)(s_n)$ is satisfiable by assumption, as $\sstrat$ is winning for \playerR{} in $\G_{post}$.
  Otherwise, the play prefix corresponds to a play that conforms to $\sstrat$ and is loosing for $\playerR{}$ in $\G_{post}$.
  But such a play does not exist by assumption that $\sstrat$ is a winning reachability strategy for $\playerR{}$ in $\G_{post}$.
	\qed
\end{proof}

\perfectpartition*
\begin{proof}
  \Cref{lem:gpost} can be adapted slightly to show that if \playerR{} has \emph{any} (not necessarily expressable as a symbolic strategy) winning strategy from $s$ in $\G_{post}$, then they have a winning stratey from $s$ in $\G$. 

	Hence, it is enough to prove the backward direction. Let us assume that there is a reachability player strategy $\sigma_R$ that is winning from $s$ in $\G$, and, for contradiction, a safety player strategy $\sigma_S$ that is winning from $s$ in $\G_{post}$. Consider strategy $\sigma'_S$ such that $\stratS'(\omega s') = \stratS(\omega s')$ if $(\safe \land \varphi)(s')$ is satisfiable, and else $\stratS'(\omega s') = \stratS''(\omega s')$, where $\stratS''$ is an arbitrary safety player strategy in $\G$. There exists a unique play $\rho = s_0s_1\ldots$ in $\G$ consistent with both $\sigma_R$ and $\sigma'_S$ with $s_0 = s$.  Since s satisfies $\post(C)[\var'/\var]$, we know that $\varphi(s_0)$ is valid (i.e., $\rho$ starts in $\G_{post}$), and since $\varphi$ perfectly partitions $\G$, we know that there is no $0 \leq i$ such that $\left((\reach \lor \safe) \land \varphi \land \neg \varphi'\right)(s_i, s_{i+1})$ holds. It follows that for all $0 \leq j$ we have that $\varphi(s_j)$ holds (i.e., $\rho$ stays in $\G_{post}$). Hence, $\rho$ is a play in $\G_{post}$ consistent with $\sigma_S$ and $\sigma_R$ (considered as a strategy in $\G_{post}$). Since $\sigma_R$ is winning in $\G$, there must exist some $n\in\mathbb{N}$ such that $\goal(s_n)$ is valid. This is a contradiction with $\sigma_S$ being a winning strategy in $\G_{post}$.
	\qed
\end{proof}


\begin{lemma}
	\label{lem:unsat_sum}
	Let $T,F$ be transition predicates such that $\unsat(\pre(T)\land\pre(F))$ holds. Suppose that $\G_T$ and $\G_F$ are games such that  $\unsat(\enf(T,\G_T))$ and $\unsat(\enf(F,\G_F))$ both evaluate to true, and we have 
	\begin{gather*}
	T \implies (\safe_T \lor \reach_T) \implies (\safe \lor \reach),\\
	F \implies (\safe_F \lor \reach_F) \implies (\safe \lor \reach).
	\end{gather*}
	Then it follows that $\unsat(\enf(T \lor F,\G))$ is also true.
\end{lemma}
\begin{proof}
	We proceed by contradiction. Suppose there were $s,v \in \S$ such that $\enf(T \lor F,\G)(s,v)$ holds. It follows that $(T \lor F)(s,v)$ must hold. Since we have $\unsat(\pre(T)\land\pre(F))$, two cases may occur: either (1) $\pre(T)(s)$ holds, or (2) $\pre(F)(s)$ holds, but not both. Since they work completely symmetrically, we give the argument only for case (2).
	
	So we know that $F(s,v)$ holds and $T(s,v)$ does not. As $F \implies (\safe_F \lor Reach_F)$, we know $(\safe_F \lor Reach_F)(s,v)$ holds. The only way we can still have $\unsat(\enf(F,\G_F))$ is if there is some $s'$ such that $\big(\safe_F(s,s')\land\lnot F(s,s')\big)$ holds. However, as $(\safe_F \lor \reach_F) \implies (\safe \lor \reach)$, then also $\big(\safe(s,s')\land\lnot F(s,s')\big)$, which is a contradiction to the fact that $\enf(F \lor T,\G)(s,v)$ holds. 
	\qed
\end{proof}

\partialcorrectness*
\begin{proof}
	We show by induction on the recursion depth that if $\algreach(\G)$ returns $(R,\sstrat_{\playerR{}},\sstrat_{\playerS{}})$, where $R \in \logic(\var)$, $\sstrat_{\playerR}$ and $\sstrat_{\playerS}$ are symbolic strategies for \playerR{}/\playerS{} respectively, and
	\begin{enumerate}
		\item[(a)] $\sstrat_{\playerR{}}$  is a winning strategy for \playerR{} from all states satisfying $R \lor \pre(\sstrat_{\playerR{}})$,
		\item[(b)] $\lnot \sstrat_{\playerS{}}$  is a necessary subgoal in $\G^I = \langle \init \land \lnot \mathit{R}, \safe, \reach, \goal \rangle$, and
		\item[(c)] if $\init \land \neg R$ is satisfiable, then $\unsat(\enf(\lnot \sstrat_{\playerS{}},\G))$ holds.
	\end{enumerate}
	Note that the last condition is equivalent to $\unsat(\enf(\lnot \sstrat_{\playerS{}},\G^I))$ as $\enf$ does not take initial states into account. Properties (b) and (c) ensure, in view of \Cref{lem:avoidingC}, that $\sstrat_{\playerS{}}$ is indeed a winning symbolic strategy from all initial states of $\G^I$ (provided that $\sstrat_{\playerS} \implies (\safe \lor \reach)$ holds, which we will show is the case). Together, these properties ensure that $R$ characterizes the set of initial states that are winning for $\playerR$ in $\G$. We distinguish the five cases that can occur when the algorithm terminates (there are four \textbf{return} statements, the last of which depends on the \textbf{if} statement in line~\ref{line:transback}).
	
	\medskip
	{\it Case 1:  $\algreach(\G)$ returns in line~\ref{line:ret1}}. Then the formula $(\init \implies \goal)$ is valid, which means that all initial states are goal states.
	The algorithm returns $R = \init\land\goal$, which in this case is equivalent to $\init$.
	Trivially, any strategy is winning for \playerR{} for all initial states.
  Then, $\sstrat_{\playerR} = \false$ is a winning symbolic strategy for any state satisfying $R$.
  This is because clearly any play starting in $R$ is winning for \playerR{}, and we only need to show that \playerR{} has a move satisfying $\sstrat_{\playerR}$ in a play prefix, which by definition has not seen $\goal$ already.
  In this case, there is no such play prefix that we need to cover.
  Observe also that $\pre(\false) \equiv \false$, which is required to show (a).

 As $\G^I$ has no initial states, one can choose $\sstrat_{\playerS{}}$ arbitrarily, as anything qualifies as necessary subgoal if $I \equiv \false$, in which case (c) is also directly satisfied. So the properties (a)--(c) above are satisfied.
	
	\medskip
	{\it Case 2:  $\algreach(\G)$ returns in line~\ref{line:safety wins}}. Then $\enf(C,\G)$ is unsatisfiable, where $C$ is the instantiation of the interpolant $\varphi$ (lines~\ref{line: phi} and \ref{line: subgoal}). 
 	\Cref{prop_necessary} states that $C$ is a necessary subgoal in $\G^I$. Since $\enf(C,\G)$ is unsatisfiable, \Cref{lem:avoidingC} states that $\sstrat_{\playerS{}} = \safe \land \lnot C$ is a winning symbolic strategy for $\playerS$ in $\G^I$. For \playerR{}, we have the same argument as in Case 1. We conclude that the properties (a)--(c) above are satisfied.
	
	\medskip
	{\it Case 3:  $\algreach(\G)$ returns in line~\ref{line:thirdreturn}}. By induction hypothesis we assume that the recursive call in line~\ref{line: recursion1} returned a tuple $(R_{post}, \sstrat_{\playerR{}}^{post}, \sstrat_{\playerS{}}^{post})$ satisfying properties (a)--(c) above for $\G_{post}$. We now show these properties in $\G$ for $R = \init\land\goal$, $\sstrat_{\playerR{}} = \false$, and $\sstrat_{\playerS{}} = \safe \land \neg F  \land (\varphi \implies \sstrat_{\playerS{}}^{post})$, where $F= C \land \mathit{R}_{post}\subsp$.
	
	As in the previous cases, property (a) is satisfied. In order to prove (b) let $\rho = s_0 s_1 \ldots$ be a play such that $(\init\land\lnot R)(s_0)$ and $\goal(s_n)$ hold for some $n \in \mathbb{N}$. As $C$ is a necessary subgoal in $\G$, there exists $m \in \mathbb{N}$ with $m < n$ such that $C(s_m,s_{m+1})$ holds. Take $m$ to be the last index with this property.  If $F(s_m,s_{m+1})$ holds, then also $(F  \lor (\varphi \land \lnot\sstrat_{\playerS{}}^{post}))(s_m,s_{m+1})$, which implies $\lnot\sstrat_{\playerS{}}(s_m, s_{m+1})$. If $\lnot F(s_m,s_{m+1})$ holds, then $(\post(C)\subsm \land \lnot R_{post})(s_{k+1})$ holds. By induction hypothesis, all winning plays for $\playerR$ in $\G_{post}$ starting in $\lnot R_{post}$ have some $n > j > m$ such that $\lnot\sstrat_{\playerS{}}^{post}(s_j,s_{j+1})$ holds.
  Furthermore, by construction of $\G_{post}$, all states $s_j$ with $n > j > m$ satisfy $\varphi$.
  Hence $\varphi \land \lnot\sstrat^{post}_{\playerS{}}(s_j,s_{j+1})$ holds, which implies that $\lnot \sstrat_{\playerS}(s_j,s_{j+1})$ holds. We conclude that $\lnot\sstrat_{\playerS{}}$ is a necessary subgoal in $\G^I$.
	
	We now show (c). Since we return in  line~\ref{line:thirdreturn}, we have $\unsat(\enf(F,\G))$ and, by induction hypothesis, $\unsat(\enf(\lnot\sstrat_{\playerS{}}^{post},\G_{post}))$.
  As the transition relation of $\G_{post}$ is restricted to $\varphi$, this implies $\unsat(\enf(\varphi \land \lnot\sstrat_{\playerS{}}^{post},\G_{post}))$.
 We also have $F \implies \lnot \varphi$ and $(\varphi \land \lnot\sstrat_{\playerS{}}^{post}) \implies \varphi$. As $(\safe_{post} \lor \reach_{post}) \implies (\safe \lor \reach)$ holds, we can apply \Cref{lem:unsat_sum} to conclude $\unsat(\enf(F  \lor (\varphi\land\lnot\sstrat_{\playerS{}}^{post})),\G)$, which implies $\unsat(\enf(\neg \safe \lor F  \lor (\varphi\land\lnot\sstrat_{\playerS{}}^{post})),\G) = \unsat(\enf(\lnot\sstrat_{\playerS{}}),\G))$.
		
	\medskip
	{\it Case 4:  $\algreach(\G)$ returns in line~\ref{line: last return} and the {\upshape\textbf{if}} statement in line~\ref{line:transback} is false}.
	By induction hypothesis we assume that the recursive calls in lines~\ref{line: recursion1} and \ref{line: recursion2} returned tuples $(R_{post}, \sstrat_{\playerR{}}^{post}, \sstrat_{\playerS{}}^{post})$ and $(R_{pre}, \sstrat_{\playerR{}}^{pre}, \sstrat_{\playerS{}}^{pre})$ satisfying properties (a)--(c) above for $\G_{post}$ and $\G_{pre}$. We now show these properties in $\G$ for $R\lor R_{pre}$, and 
	\begin{align*}
		\sstrat_{\playerR{}} &=  \;\;(\pre(\sstrat_{\playerR{}}^{post}) \implies \sstrat_{\playerR{}}^{post}) \\
  & \land \;\; ((\neg \pre(\sstrat_{\playerR{}}^{post}) \land \pre(F)) \implies F) \\
  & \land \;\; ((\neg \pre(\sstrat_{\playerR{}}^{post}) \land \neg \pre(F)) \implies \sstrat_{\playerR{}}^{pre}) \\
    & \land \;\; (\pre(\sstrat_{\playerR{}}^{post}) \lor \pre(F) \lor \pre(\sstrat_{\playerR{}}^{pre})) \\
		\sstrat_{\playerS{}} &= (\neg \varphi \implies \sstrat_{\playerS{}}^{pre}) \land (\varphi \implies \sstrat_{\playerS{}}^{post}),
	\end{align*} 
	with $F= C \land \mathit{R}_{post}\subsp$.

We first show that (a) $\sstrat_{\playerR{}}$ is winning for \playerR{} from states satisfying $R\lor R_{pre} \lor \pre(\sstrat_{\playerR})$. For states in $R$ this is trivial, so let  $\rho = s_0 s_1 \ldots$ be a play in $\G$ conforming to $\sstrat_{\playerR{}}$ such that $R_{pre}(s_0)$ holds.
  Our first claim is that if there exists $k \in \mathbb{N}$ such that $\pre(\sstrat_{\playerR}^{post})(s_k)$ holds, then $\rho$ must be winning for \playerR{}.
  This is due to the fact that $\sstrat_{\playerR}^{post}$ is winning in $\G_{post}$ from all states satisfying $\pre(\sstrat_{\playerR}^{post})$, which allows us to use~\Cref{lem:gpost}.
  To argue that $\sstrat_{\playerR}^{post}$ keeps playing according to $\sstrat_{\playerR}^{post}$ once such a state is reached, we observe that if a symbolic reachability strategy $\sstrat$ wins from $s$, then $\pre(\sstrat)$ holds in any state in $S_{\playerR}$ reachable from $s$ via a play prefix conforming to $\sstrat$, by definition.
  
  Now we show that such a position $k$ must exist.
  First, for $j \in \mathbb{N}$ such that $(\neg \pre(\sstrat_{\playerR}^{post}) \land \pre(\enf(F,\G)))(s_j)$ holds, the transition $(s_j,s_{j+1})$ must satisfy $F$.
  This is because if $s_j \in S_{\playerS}$, then all outgoing transitions from $s_j$ satisfy $F$.
  Otherwise, it follows by the fact that $\rho$ conforms to $\sstrat_{\playerR}$.
  As $\post(F) \equiv R_{post}\subsp$ and $\sstrat_{\playerR}^{pos}$ wins from all states satisfying $R_{post}$ by assumption, it follows that $s_{j+1}$ satisfies $\pre(\sstrat_{\playerR}^{post})$.

  As long as $\rho$ visits only states satisfying $(\neg \pre(\enf(F,\G)) \land \neg \pre(\sstrat_{\playerR}^{post}))$, the strategy $\sstrat_{\playerR}$ prescribes to play according to $\sstrat_{\playerR}^{pre}$.
  By assumption, this strategy is winning for \playerR{} in $\G_{pre}$, and hence the play $\rho$ eventually visits a state in $\pre(\enf(F,\G))$.
  As above, the play is guaranteed to stay in $\pre(\sstrat_{\playerR}^{pre})$ until that position.

  The above argument also shows that $\sstrat_{\playerR}$ is winning for all states satisfying $\pre(F) \lor \pre(\sstrat_{\playerR}^{post}) \lor \pre(\sstrat_{\playerR}^{pre})$, which is implied by $\pre(\sstrat_{\playerR})$.
  Also, $\sstrat_{\playerR} \implies (\reach \lor \safe)$ is valid, as the corresponding statements hold for the pre- and post-strategies, and $F \implies (\safe \lor \reach)$ is valid.

	
	 
	
	Next we show that (b) $\lnot \sstrat_{\playerS{}}$  is a necessary subgoal in $\G^I$.
	No player can play back from $\G_{post}$ to $\G_{pre}$ without $\playerR$ having already won in $\G_{post}$. We first show that under this condition,
  \[\neg \sstrat_{\playerS} = (\neg \varphi \land \neg \sstrat_{\playerS{}}^{pre}) \lor (\varphi \land \neg \sstrat_{\playerS{}}^{post})\]
  qualifies as a necessary subgoal in $\G^I$. For this, consider the necessary subgoal $C$.
  For any play $\rho = s_0 s_1 \ldots$ with $n \in \mathbb{N}$ such that $\goal(s_n)$ there is some $k \in \mathbb{N}$ with $k < n$ and $C(s_k,s_{k+1})$. As $F$ characterizes a subset of $C$, we check two cases: Either (1) $\lnot F(s_k,s_{k+1})$ or (2) $F(s_k,s_{k+1})$. In case (1), we have $\lnot R_{post}(s_{k+1})$ and because of our assumption that no transition of the game satisfies $\varphi \land \neg \varphi'$, for all $j \in \mathbb{N}$ with $k < j < n: \varphi(s_j)$. It follows by induction hypothesis that there is some $l \in \mathbb{N}$ such that $\neg \sstrat_{\playerS}^{post}(s_l,s_{l+1})$.
  In case (2) we use that $\sstrat_{\playerS}^{pre}$ plays only moves available in $\G_{pre}$, and hence $\sstrat_{\playerS}^{pre} \implies \neg F$ is valid.
  Furthermore $F \implies \neg \varphi$, as $F$ characterizes a subset of the subgoal $C$.
  Hence we can conclude that $F \implies (\neg \varphi \land \neg \sstrat_{\playerS}^{pre})$ is valid.
  It follows that $\neg \sstrat_{\playerS}$ qualifies as necessary subgoal.

  Finally we show (c) that $\unsat(\enf(\neg \sstrat_{\playerS},\G))$ holds.
  We have $(\safe_{pre} \lor \reach_{pre}) \implies (\safe \lor \reach)$ and $(\safe_{post} \lor \reach_{post}) \implies (\safe \lor \reach)$. As $\pre(\neg \varphi \land \neg \sstrat_{\playerS}^{pre}) \land \pre(\varphi \land \neg \sstrat_{\playerS}^{post})$ is cleary unsatisfiable, we can again apply Lemma \ref{lem:unsat_sum} to infer $\unsat(\enf(\neg \sstrat_{\playerS},\G))$.
  This uses that any transitions reachable in $\G_{post}$ has to satisfy $\varphi$ in this case.
	
	\medskip
	{\it Case 5:  $\algreach(\G)$ returns in line~\ref{line: last return} and the {\upshape\textbf{if}} statement in line~\ref{line:transback} is true}. 

	We  assume that both recursive calls terminated and, by induction, returned triples $(R_{post},\sstrat_{\playerR}^{post},\sstrat_{\playerS}^{post})$ and $(R_{pre},\sstrat_{\playerR}^{pre},\sstrat_{\playerS}^{pre})$ satisfying (a)-(c).

  (a) is shown exactly as in Case 4.

	For (b) we first observe that by setting $\varphi$ to $\false$ (see line~\ref{line:tpostfalse}) in this case we get $\sstrat_{\playerS} = \sstrat_{\playerS}^{pre}$.
  We show that $\neg \sstrat_{\playerS} = \neg \sstrat_{\playerS}^{pre}$ is a necessary subgoal in $\G_I$.
  The transition predicate $F$ in line~\ref{line: recursion1} is a sufficient subgoal by induction hypothesis, but due to the restriction on the post-game, we cannot conclude that states in $\post(C)$ that are not in $\post(F)$ are winning for \playerS{}.
	By adding all transitions to $\goal$ (line~\ref{line: E2}) we get that $F$ in line~\ref{line: recursion2} is a necessary and sufficient subgoal (clearly, any winning play must go through $\goal \subsp$).
	As we have ensured that $F$ is necessary, we know for all plays $\rho = s_0 s_1 \ldots$ with some $n \in \mathbb{N}$ such that $\goal(s_n)$ there is some $k \in \mathbb{N}$ with $k < n$ and $F(s_k,s_{k+1})$. As in Case 4 we may conclude that $F \implies \neg \sstrat_{\playerS}^{pre}$.
It follows that $\neg \sstrat_{\playerS}$ is a necessary subgoal in $\G_I$.

For (c) we observe that $\unsat(\pre(\enf(\neg \sstrat_{\playerS}^{pre},\G_{pre})))$ holds by induction hypothesis, which directly implies $\unsat(\pre(\enf(\neg \sstrat_{\playerS}^{pre},\G)))$. This concludes the argument for the final case, and the proof is complete.
	\qed
\end{proof}

\terminationfinite*
\begin{proof}
	We denote by $\size(\G)$ the number of concrete transitions of $\G$, formally: $\size(\G) = |\{ (s,s') \in S \times S \mid (\safe \lor \reach)(s,s') \text{ is valid}\}|$.
	If the domains of all variables are finite, then so is $\size(\G)$.
	We assume that this is the case and show that the subgames on which $\algreach(\G)$ recurses are strictly smaller in this measure.
	This is enough to guarantee termination.
	
	The first subgame is constructed in line~\ref{line:gpost} and takes the form:
	\[\G_{post} = \langle \post(\mathit{C})\subsm,\safe \land \varphi, \reach \land \varphi, \goal \rangle.\]
	The important restriction of this game is that both safety and reachability player transitions have the additional precondition $\varphi$.
	We may assume that $\enf(C,\G)$ is satisfiable, as otherwise the algorithm does not reach line~\ref{line:gpost}.
	Then, in particular, $C$ is satisfiable, by the definition of $\enf(C,\G)$.
	But $C = \instantiate(\varphi,\G) = (\safe \lor \reach) \land \neg \varphi \land \varphi'$, which means that there exist states $s,s'$ such that $ (\safe \lor \reach)(s,s')$, $\neg \varphi(s)$, and $\varphi(s')$ are all valid. 
	This transition from $s$ to $s'$ in $\G$ is excluded in $\G_{post}$, and as no new transitions are included, it follows that $\size(\G_{post}) < \size(\G)$.
	
	The second subgame is constructed in line~\ref{line:gpre2} and takes the form:
	\[\G_{pre} = \langle I,\safe \land \lnot F,\reach \land \lnot F, \pre(\enf(F,\G))\rangle.\]
  We may assume that $F \land (\safe \lor \reach)$ is satisfiable, as otherwise the algorithm would not have moved past line~\ref{line:safeavoidsE}.
  Observe that if $F$ is changed in line~\ref{line: E2} then it is only extended and hence satisfiability is preserved.
  As no transition satisfying $F$ exists in $\G_{pre}$ it follows that $\size(\G_{pre}) < \size(\G)$.
  This concludes the proof.
	\qed
\end{proof}

\terminationbisim*
\begin{proof}
	Let $S_1,\ldots,S_n$ be the bisimulation classes of $\G$, and $\psi_1, \ldots, \psi_n \in \logic(\var)$ be the formulas that define them.
	We define
	\[\size(\G) = |\{(S_i,S_j) \mid (\safe \lor \reach) \land \psi_i \land \psi_j' \text{ is satisfiable} \}|,\]
	which equals the number of transitions in the bisimulation quotient of $\G$ under $\sim$.
	Our aim is to show that $\algreach(\cdot)$ terminates for all subgames that are considered in any recursive call of $\algreach(\G)$.
	
	To this end, we show that $\algreach(\G)$ terminates for all reachability games $\G = \langle \init, \safe, \reach, \goal \rangle$ such that
	\begin{itemize}
	\item $\size(\G)$ is finite,
  \item the relation $\sim$ is a bisimulation on $\G$, and
	\item $\goal$ is equivalent to a disjunction of formulas $\psi_i$.
	\end{itemize}
	We show this by induction on $\size(\G)$.
	
	Let $\G = \langle \init, \safe, \reach, \goal \rangle$ satisfy these conditions, and assume that $\size(\G) = 0$.
	Then it follows that $\safe \lor \reach$ is unsatisfiable.
	This is because if any $(s_1,s_2)$ would satisfy $\safe \lor \reach$, then in particular $(\safe \lor \reach) \land \psi_i \land \psi_j'$ would be satisfied by $(s_1,s_2)$, where we assume $s_1 \in \S_i$ and $s_2 \in \S_j$.
	It follows that $\unsat(\enf(C,\G))$ in line~\ref{line: cpre} is true, as $\enf(C,\G) \implies (\safe \lor \reach)$ is valid for any $C$.
	But then Algorithm~\ref{alg:algreach} terminates on input $\G$.
	
	Now suppose that we have $\G$ with $\size(\G) > 0$.
	If the algorithm does not return in lines~\ref{line:ret1} or~\ref{line:safety wins}, we have to consider the first subgame
	\[\G_{post} = \langle \post(C)\subsm, \safe \land \varphi, \reach \land \varphi, \goal \rangle,\]
	which is constructed in line~\ref{line:gpost}.
	We may assume that for some $I \subseteq \{1,\ldots,n\}$ we have $\varphi \equiv \bigvee_{i \in I} \psi_i$, due to our assumption on the function $\interpolate$.
	Hence the effect of restricting all transitions to $\varphi$ is to remove all transitions in states not in $\bigcup \{S_i \mid i \in I\}$, which are exactly the states in $\bigcup \{S_i \mid i \in \{1,\ldots,n\} \setminus I\}$.
	It is clear that $\sim$ is still a bisimulation in the resulting game, and that the goal states are preserved.
	To see that $\size(\G_{post}) < \size(\G)$ we may assume that $\unsat(\enf(C,\G))$ is false, otherwise we would have returned in line~\ref{line:safety wins}.
	Then, in particular, there is a transition in $\G$ satisfying $\neg \varphi$, which means that there is a pair $S_i,S_j$ such that $(\safe \lor \reach) \land \neg \varphi \land \psi_i \land \psi_j'$ is satisfiable.
	This is cleary unsatisfiable when replacing $(\safe \lor \reach)$ by $(\safe \lor \reach) \land \varphi$.
	Hence, $\size(\G_{post}) < \size(\G)$.
	As a result, we can apply the induction hypothesis to conclude that the recursive call $\algreach(\G_{post})$ in line~\ref{line: recursion1} terminates.
	
	Now let us consider the second subgame $\G_{pre}$, as constructed in line~\ref{line:gpre2}.
	First, we observe that $F \equiv (\safe \lor \reach) \land \neg \varphi \land \varphi' \land R_{post}\subsp$, where $R_{post}$ is a state predicate characterizing the initial winning states of $\G_{post}$ (this uses~\Cref{thm:partcorr}).
	As $\sim$ is a bisimulation on $\G_{post}$, it follows by~\Cref{lem:bisimpreservesreach} that $R_{post}$ is equivalent to a disjunction of formulas $\psi_i$.
	As a consequence, we can equivalently write $F$ as $\phi_1 \land (\phi_2\subsp)$ for two formulas $\phi_1,\phi_2 \in \logic(\var)$ that are both equivalent to disjunctions of $\psi_i$.
	By~\Cref{lem:preenfbisim} it follows that $\pre(\enf(E,\G))$ is also equivalent to a disjunction of $\psi_i$.
	
	Restricting transitions to $\neg F$ in $\G_{pre}$ has the effect of removing all transitions from states in $\bigcup\{S_i \mid \psi_i \implies \phi_1 \text{ is valid}\}$ to states in $\bigcup\{S_i \mid \psi_i \implies \phi_2  \text{ is valid}\}$.
	It is clear that $\sim$ is still a bisimulation in the resulting game.
	Furthermore, as $\enf(F,\G)$ is satisfiable, there is at least one such transition in $\G$.
	It follows that $\size(\G_{pre}) < \size(\G)$ and hence the algorithm terminates by induction hypothesis.
	\qed
\end{proof}

\begin{lemma}
	\label{lem:preenfbisim}
	Let $\sim$ be a bisimulation on $\G$ which is also an equivalence relation, and $S_1,\ldots, S_n$ be its equivalence classes.
	Assume that $S_1,\ldots, S_n$ are defined by $\psi_1,\ldots, \psi_n \in \logic(\var)$.
	Let $\phi_1 \land (\phi_2\subsp) \in \logic(\var \cup \varp)$ be such that both $\phi_1,\phi_2$ are equivalent to disjunctions of formulas $\psi_i$.
	
	Then, $\pre(\enf(\phi_1 \land (\phi_2\subsp),\G))$ is equivalent to a disjunction of formulas $\psi_i$.
\end{lemma}
\begin{proof}
	We show that if there exists a state in $S_i$ that satisfies $\pre(\enf(\phi_1 \land (\phi_2\subsp),\G))$, then so do all states in $S_i$.
	Let $s_1 \in S_i$ be such that $\pre(\enf(\phi_1 \land (\phi_2\subsp),\G))(s_1)$ is valid.
	
	We make a case distinction on whether $s_1 \in S_{\playerR}$ holds.
	If so, then there exists a state $q_1$ such that $(\reach \land \phi_1 \land (\phi_2\subsp))(s_1,q_1)$ is valid.
	In particular, $\phi_1(s_1)$ and $\phi_2(q_1)$ are both valid.
	Assuming that $q_1 \in S_j$ holds, both $\psi_i \implies \phi_1$ and $\psi_j \implies \phi_2$ are valid, as both $\phi_1$ and $\phi_2$ are equivalent to disjunctions of $\psi$-formulas (which have pairwise disjoint sets of models).
	Now take any other state $s_2 \in S_i$.
	As $s_1 \sim s_2$ and $\reach(s_1,q_1)$ is valid, there exists a state $q_2 \in S_j$ such that $\reach(s_2,q_2)$ is valid.
	Furthermore, as $\psi_i(s_2)$ and $\psi_j(q_2)$ are both valid, so is $(\reach \land \phi_1 \land (\phi_2\subsp))(s_2,q_2)$.
	Hence, $\pre(\enf(\phi_1 \land (\phi_2\subsp),\G))(s_2)$ is valid.
	
	Now assume that $s_1 \in S_{\playerS}$.
	Then, for all states $q_1$ such that $\safe(s_1,q_1)$ is valid, $(\phi_1 \land (\phi_2\subsp))(s_1,q_1)$ holds.
	Whenever this is the case, and $q_1 \in S_j$ holds, it follows that $\psi_j \implies \phi_2$ is valid.

	Now take any other state $s_2 \in S_i$ and assume, for contradiction, that there exists a $q_2$ such that $\safe(s_2,q_2)$ is valid, but not $(\phi_1 \land (\phi_2\subsp))(s_2,q_2)$.
	Assuming $q_2 \in S_j$, we have that $\psi_j \land \phi_2$ is unsatisfiable.
	As $s_1 \sim s_2$ holds, we find $q_1$ such that $\safe(s_1,q_1) \land \psi_j(q_1)$ is valid.
  By the previous reasoning, this would imply that $\psi_j \implies \phi_2$ is valid.
	This is a contradiction as $\psi_j$ is satisfiable.
	\qed
\end{proof}

\end{document}